\documentclass[12pt,a4paper]{amsart}

\usepackage{tikz}

\usepackage{amsmath,amsfonts,amssymb}
\usepackage{mathrsfs}  

\usepackage[hmargin=2cm,vmargin=2cm]{geometry}

\usepackage{nameref,zref-xr}

\newcommand{\mbZ}{\mathbb Z}
\newcommand{\mbC}{\mathbb C}

\newcommand{\oM}{\overline{\mathcal M}}

\newcommand{\tu}{{\widetilde u}}
\newcommand{\og}{\overline g}
\newcommand{\oh}{\overline h}
\newcommand{\hLambda}{\widehat\Lambda}

\def\oM{{\overline{\mathcal{M}}}}

\def\d{{\partial}}

\renewcommand{\>}{\right>}
\newcommand{\eps}{\varepsilon}

\newcommand{\hcA}{\widehat{\mathcal{A}} }
\newcommand{\DR}{\mathrm{DR}}
\newcommand{\DZ}{\mathrm{DZ}}

\newcommand{\Coef}{\mathrm{Coef}}

\newcommand{\oG}{{\overline G}}
\newcommand{\oH}{{\overline H}}
\newcommand{\tC}{\widetilde C}
\newcommand{\Li}{\mathrm{Li}}
\newcommand{\cD}{{\mathcal{D}}}

\newcommand{\bfr}{\mathfrak{b}}
\newcommand{\nfr}{\mathfrak{n}}

\newcommand{\res}{\mathrm{res}}
\newcommand{\Dcl}{\mathcal{D}}
\newcommand{\Lsc}{\mathscr{L}}
\newcommand{\ts}{\textstyle}

\newtheorem{theorem}{Theorem}[section]
\newtheorem{proposition}[theorem]{Proposition}

\newtheorem{corollary}[theorem]{Corollary}

\theoremstyle{definition}

\usepackage{color}

\numberwithin{equation}{section}

\title{Quantum ${D_4}$ Drinfeld-Sokolov hierarchy and quantum singularity theory}

\author{Ann du Crest de Villeneuve}
\address{A.~du Crest de Villeneuve:\newline LAREMA, UMR6093 CNRS, Universit\'e d'Angers, 49000 Angers, France}
\email{ducrest@math.univ-angers.fr}

\author{Paolo Rossi}
\address{P.~Rossi:\newline Dipartimento di Matematica ``Tullio Levi-Civita'', Universit\`a degli Studi di Padova,\newline
Via Trieste 63, 35121 Padova, Italy}
\email{paolo.rossi@math.unipd.it}

\begin{document}

\begin{abstract}
In this paper we compute explicitly the double ramification hierarchy and its quantization for the $D_4$ Dubrovin-Saito cohomological field theory obtained applying the Givental-Teleman reconstruction theorem to the $D_4$ Coxeter group Frobenius manifold, or equivalently the $D_4$ Fan-Jarvis-Ruan-Witten cohomological field theory  (with respect to the non-maximal diagonal symmetry group $\<J\> = \mbZ/3\mbZ$). We then prove its equivalence to the corresponding Dubrovin-Zhang hierarchy, which was known to coincide with the $D_4$ Drinfeld-Sokolov hierarchy. Our techniques provide hence an explicit quantization of the $D_4$ Drinfeld-Sokolov hierarchy. Moreover, since the DR hierarchy is well defined for partial CohFTs too, our approach immediately computes the DR hierarchies associated to the invariant sectors of the $D_4$ CohFT with respect to folding of the Dynkin diagram, the $B_3$ and $G_2$ Drinfeld-Sokolov hierarchies.
\end{abstract}

\markboth{A. du Crest de Villeneuve, P. Rossi}{Quantization of the $D_4$ Drinfeld-Sokolov hierarchy and FJRW theory}

\maketitle

\tableofcontents

\section*{Introduction}

In \cite{Dub99}, based on the results of K. Saito \cite{Sai81,Sai83a,Sai83b} and in particular on his theory of primitive forms, Dubrovin constructed a generically semisimple Frobenius manifold structure on the space of orbits of finite, irreducible Coxeter groups or, via the relation of the latter with simple hypersurface singularities, on the space of miniversal deformations of any polynomial $W:\mbC^n\to\mbC$ with an isolated critical point at the origin.\\

Although in general these Frobenius manifolds are not semisimple at the origin, it was proved in \cite{Mil14} that, for singularities of type $ADE$, using Givental-Teleman theory \cite{Giv01,Tel12} at a nearby semisimple point and then shifting the result back to the origin yields a well defined conformal cohomological field theory \cite{KM94} which we refer to as the Dubrovin-Saito CohFT of the singularity.\\

In \cite{FJR07,FJR13} Fan, Jarvis and Ruan, inspired by ideas of Witten, introduced another construction associating a certain moduli space of decorated curves and, through a virtual fundamental class on it, a conformal CohFT to a certain class of quasi-homogeneous polynomials with an isolated critical point at the origin, together with the choice of admissible symmetry group.\\

In case the symmetry group is the full automorphism group $G_\mathrm{max}$ of the polynomial, in a typical instance of mirror symmetry, there is an isomorphism between the FJRW CohFT and the Dubrovin-Saito CohFT of the transposed singularity \cite{FJR13}. When the symmetry group is smaller the mirror partner to the FJRW theory is in general not known, but for the case of $ADE$ singularities the only admissible groups are always maximal, with the exception of singularities of type $D_{2n}$ for which however the isomorphism with the Dubrovin-Saito theory still holds.\\

The case of $D_4$ is in many ways the most subtle among the $ADE$ singularities. For instance the general method of proof for the mirror symmetry result of \cite{FJR13} did not work in the $D_4$ case because of the specific form of the CohFT and its phase space, and indeed the proof for the mirror theorem was completed in \cite{FFJMR10}. In essence the complication (specifically the appearance of the so called broad sectors of the phase space) originates from the peculiar symmetry of the $D_4$ singularity. This is apparent from the corresponding Dynkin diagram, the most symmetric among the $ADE$ diagrams.\\

In this paper we focus on the example of the $D_4$ Dubrovin-Saito CohFT, or the FJRW CohFT of the singularity $D_4:W=x^3+xy^2$ with symmetry group $\langle J_1 \rangle$ where $J_1(x,y) = (e^{2 \pi i \frac{1}{3}}x,e^{2 \pi i \frac{1}{3}}y)$ and approach it from the point of view of the associated integrable hierarchies.\\

It is well known that, given a semisimple CohFT one can associate to it the Dubrovin-Zhang hierarchy \cite{DZ05}, an integrable hierarchy of Hamiltonian PDEs which controls the intersection theory of the CohFT with psi classes on the moduli space of stable curves. It was proved in \cite{GM05,FGM10} that the Dubrovin-Zhang hierarchy of the $ADE$ CohFTs are the Drinfeld-Sokolov hierarchies associated to the corresponding $ADE$ semisimple algebras.\\

The double ramification hierarchy is a more recent construction, introduced in \cite{Bur15,BR16b}, associating a quantum integrable system to a CohFT (there is no need for semisimplicity in this case, and in fact the construction even works in the classical limit for partial CohFTs).\\

This construction, inspired by ideas from symplectic field theory \cite{EGH00,FR10,Ros10}, uses intersection theory of the CohFT with the double ramification cycle, Hodge classes and psi classes, and its relation with the Dubrovin-Zhang hierarchy is the object of a conjecture (the strong DR/DZ equivalence conjecture of \cite{BDGR18,BGR17}) which predicts that, given a semisimple CohFT (so that the DZ side is well defined), the classical limit of the DR hierarchy and the DZ hierarchy coincide after a specific change of coordinates in the phase space of the system.\\

This conjecture has been proved for several specific CohFTs \cite{BDGR18,BDGR16b,BR18} and, in particular, in \cite{BG16,BDGR18}, for the $A_N$ Dubrovin-Saito CohFTs (which are nothing but the $(N+1)$-spin CohFTs) for $N\leq 4$.\\

In this paper we show that the classical double ramification hierarchy of the $D_4$ CohFT coincides with the $D_4$ Drinfeld-Sokolov hierarchy (with no need of coordinate change in this case) and hence, via the results of \cite{GM05}, with the Duborinv-Zhang hierarchy. Moreover we compute the quantum double ramification hierarchy which hence provides an explicit (closed formulas) quantization of the $D_4$ Drinfeld-Sokolov hierarchy.\\

Finally we exploit the particularly symmetric nature of the $D_4$ CohFT, which possesses both a $\mbZ_2$ and a $\mbZ_3$ symmetry, to consider its invariant parts as two partial CohFTs, for which the classical DR hierarchy gives rise to the $B_3$ and $G_2$ Drinfeld-Sokolov hierarchies, and thanks to the results of \cite{LRZ15} we can deduce that the partial CohFT potentials are tau functions of the corresponding DR hierarchy.\\

\subsection*{Acknowledgments} We would like to express our gratitude to A. Buryak as well as to M. Cafasso and V. Roubtsov for useful discussions. A. du Crest de Villeneuve also thanks the University of Angers, France, the University of Bourgogne, France and the European Marie Skłodowska-Curie action `IPaDEGAN' for offering them the best working conditions while working on the present paper.

\section{Drinfeld-Sokolov $D_4$ hierarchy}\label{DSD4}

\noindent In \cite{DS84}, Drinfeld and Sokolov described how to associate a hierarchy of integrable PDEs to any semi-simple Lie algebra through its associated affine Kac--Moody algebra. In the most general case, the hierarchy is constructed as an infinite sequence of \textit{matrix} Lax equations (zero-curvature equations). In the same article, for simple Lie algebras of types $A_n,B_n$ and $C_n$, the authors also provided a \textit{scalar} Lax pair along with a bi-Hamiltonian representation using pseudo-differential operators. For the $D_n$ case however, Drinfeld and Sokolov represented one part of the hierarchy with a scalar Lax pair and bi-Hamiltonian representation (\cite{DS84}, pp. 2019--2021), this part we call the \textit{positive flows} of the hierarchy, while the remaining part we call the \textit{negative flows}.

In \cite{LWZ10}, the authors produced a complete picture of the hierarchy of $D_n$ type (positive and negative flows) in terms of scalar Lax pairs and bi-Hamiltonian representation. To do so, they introduced a new kind of pseudo-differential operators called \textit{of the second type} (while the traditional ones are called \textit{of the first type)}. Roughly speaking, the operators of the second type are allowed to contain not only infinitely many nonzero terms in negative degree, but also infinitely many nonzero terms in positive degrees, along with considerations of gradation to ensure the well-definedness of the product of operators. In that same article \cite{LWZ10}, the authors also described the tau-symmetric bi-Hamiltonian structure (in the sense of \cite{DZ05}) of the Drinfeld--Sokolov hierarchy of $D_n$ type.

In what follows, we briefly review the original Drinfeld--Sokolov reduction for the $D_4$ case, for it entirely defines the hierarchy \cite{DS84}. Then we follow the approach of \cite{LWZ10} to define the scalar Lax pairs of the positive and negative flows of the hierarchy. Finally, we describe and compute the tau-symmetric bi-Hamiltonian structure of $D_4$.

\subsection{Matrix Lax equations of type $D_4$}

Let us briefly explain what these positive and negative flows are in terms of Lie algebras. Type $D_4$ is the Dynkin diagram of the simple Lie algebra $\mathfrak{o}(8)$. We denote by
\begin{align*}
\mathfrak{g} = \mathfrak{o}(8)\otimes \mathbb{C}[\lambda,\lambda^{-1}]
\end{align*}
the loop algebra of $\mathfrak{o}(8)$. We choose generators $\{e_i, f_i, h_i \;|\; 0 \leq i\leq 4\}$ of $\mathfrak{g}$ as in \cite{DS84} (also as in \cite{LWZ10}). We define the principal gradation $\mathfrak{g} = \bigoplus_{k\in \mathbb{Z}} \mathfrak{g}^k$ with $\deg e_i = -\deg f_i = 1$ and $\deg h_i = 0$. We use the notation $\mathfrak{g}^{>0} = \prod_{k>0} \mathfrak{g}^k$ and similarly for $\mathfrak{g}^{<0}$. We denote $\Lambda = \sum_{i = 0}^4 e_i \in \mathfrak{g}^1$ the principal cyclic element of $\mathfrak{g}$. Then it is well known \cite{Kos59,Kac78,DS84} that the principal Heisenberg subaglebra $\mathrm{Ker} (\Lambda)$ admits the following decomposition:
\begin{align*}
\mathrm{Ker} (\Lambda) = \bigoplus_{k\in \mathbb{Z}^\mathrm{odd}} \mathbb{C}\cdot \Lambda_k \oplus \bigoplus_{k\in\mathbb{Z}^{\mathrm{odd}}} \mathbb{C}\cdot \Gamma_k, \hspace{1cm} \Lambda_k\in \mathfrak{g}^k, \Gamma_k\in \mathfrak{g}^{3k}.
\end{align*}
(See \cite{DS84} for the expression of $\Lambda_k,\Gamma_k$.) Let us denote by $\bfr$ (resp. $\nfr$) the negative Borel (resp. negative nilpotent) subalgebra of $\mathfrak{o}(8)$. The starting point of the Drinfeld--Sokolov hierarchy is a matrix-valued differential operator of the form
\begin{align}\label{eqMatrixLaxGen}
\Lsc = \partial_x + \Lambda + q(x), \hspace{1cm} q\in \mathcal{C}^\infty(\mathbb{R},\bfr).
\end{align}
For any function $S\in\mathcal{C}^\infty(\mathbb{R}, \nfr)$, the operator $\tilde{\Lsc} = \mathrm{e}^{\mathrm{ad}_S}\Lsc$ also has the form of Equation \eqref{eqMatrixLaxGen}; we say that $\Lsc$ and $\tilde{\Lsc}$ are gauge equivalent. In \cite{DS84}, the authors proved the fundamental property that there exists a (non unique) function $U\in \mathcal{C}^\infty(\mathbb{R}, \mathfrak{g}^{<0})$ such that the operator $\Lsc_0 = \mathrm{e}^{\mathrm{ad}_U}\Lsc$ has the form
\begin{align*}
\Lsc_0 = \partial_x + \Lambda + H(x), \hspace{1cm} H\in\mathcal{C}^\infty(\mathbb{R}, \mathrm{Ker} (\Lambda)\cap \mathfrak{g}^{<0}).
\end{align*}
(Note that $\Lsc$ and $\Lsc_0$ are not gauge equivalent). The Drinfeld--Sokolov hierarchy of type $D_4$ is then defined as the two sequences of matrix Lax equations:
\begin{align}\label{eqMatrixLaxEqDn_positive}
\frac{\partial\Lsc}{\partial t_k} &= \left[ -\left( \mathrm{e}^{\mathrm{ad}_U} \Lambda_k \right)^{\geq 0}, \Lsc \right],\\
\frac{\partial\Lsc}{\partial \hat{t}_k} &= \left[ \left( \mathrm{e}^{\mathrm{ad}_U} \Gamma_k \right)^{\geq 0}, \Lsc \right],\label{eqMatrixLaxEqDn_negative}
\end{align}
for $k\in\mathbb{Z}^\mathrm{odd}_+$, where $(\,\cdot\,)^{\geq 0}$ denotes the projection onto the subspace $\mathfrak{g}^{\geq 0}$. These equations eventually take the form of evolutionary PDEs on the coordinates of the matrix $q$. Changing the representative of $\Lsc$ in its gauge equivalence class amounts to changing coordinates. They do not depend on the choice of the function $U$. The flows with respect to the variable $t_k$ \eqref{eqMatrixLaxEqDn_positive} are those we call the positive flows, while those with respect to the variables $\hat{t}_k$ \eqref{eqMatrixLaxEqDn_negative} are the negative flows.

\subsection{Salar Lax pairs}\label{section:positive flows}
	
We will start by describing the scalar Lax pairs of the positive flows \eqref{eqMatrixLaxEqDn_positive}. We follow the formal algebraic approach to pseudo-differential operators as in \cite{LWZ10}. We start by setting indeterminates $s^1_k,\ldots, s^4_k$, with $k\geq 0$ and denote $s^\alpha_0 = s^\alpha$. Define a formal derivation 
\begin{align*}
\d_x = \sum_{\alpha = 1}^4\sum_{k\geq 0} s^\alpha_{k+1} \frac{\partial}{\partial s^\alpha_k},
\end{align*}
so that $s^\alpha_{k+1} = \partial_x(s^\alpha_k)$ acting on the ring of \textit{generalized differential polynomials} 
\begin{align*}
\hcA_s = \mathbb{C}[[s^*]][s^*_{>0}] [[\eps]],
\end{align*}
where $s^*$ denotes $s^1,\ldots, s^4$. In what follows we will mostly drop the index in $\hcA_s$, simply writing $\hcA$ when no confusion can arise. We define a gradation $\hcA = \bigoplus_{k\geq 0} \hcA^{[k]}$ by setting $\deg s^\alpha_k = k$ and $\deg\eps = -1$. Changes of coordinates $s^\alpha\to \widetilde{s}^\alpha$ will be described by transformations of the form
\begin{align*}
\widetilde{s}^\alpha(s^*_{\geq 0};\eps) \in \hcA^{[0]}, \hspace{1cm} 0 \neq \det\left( \frac{\partial \left. \widetilde{s}^\alpha\right|_{\eps = 0}}{\partial s^\beta} \right)_{\alpha,\beta = 1,\ldots, 4},
\end{align*}
called \textit{Miura transforms} and should be seen as maps from $\hcA_s$ to $\hcA_{\widetilde s}$. We refer to \cite{Ros17} for more details on generalized differential polynomials. As usual, we define the ring of pseudo-differential operators that, here, we call \textit{of the first type} (following the vocabulary of \cite{LWZ10}),
\begin{align}
\mathcal{D}^- = \left\{\left. \sum_{k = -\infty}^m a^k\partial_x^k \; \right| \; m\in \mathbb{Z}, a^k\in \hcA^{[0]} \right\}.
\end{align}
The product of two pseudo-differential operators of the first type is defined by, for any $a,b\in \mathcal{D}^-$ and any $n,m\in \mathbb{Z}$,
\begin{align}\label{eqProductOfOperators}
a\partial_x^n\circ b\partial_x^m = \sum_{k\geq 0} \binom{n}{k}\, a b_k \eps^k \partial_x^{n+m-k}, & & \binom{n}{k} = \frac{n(n-1)\cdots (n-k+1)}{k!},
\end{align}
where $b_k = \partial_x^k(b)$, and extended by linearity. Notice that we added a factor $\eps^k$ to be coherent with the introduction of generalized differential polynomials. For any operator $X = \sum_{k\leq m} a^k\partial_x^k$, we call the positive part (respectively the negative part) of $X$ the operator $X_+ = \sum_{k\geq 0} a^k\partial_x^k$ (respectively $X_- = \sum_{k<0} a^k\partial_x^k$).

As already defined by Drinfeld and Sokolov \cite{DS84}, the scalar Lax operator of the hierarchy of type $D_4$ has the form\footnote{Since the coordinate $s^4$ is treated differently then the others, we use letters $\mu, \nu$ for coordinates $s^1$, $s^2$, $s^3$ and letters $\alpha, \beta$ for coordinates $s^1$, $s^2$, $s^3$ and $s^4$.}
\begin{align}\label{eqLaxOperator}
L = \partial_x^{6} + \partial_x^{-1}\sum_{\mu = 1}^{3} \left(s^\mu\partial_x^{2\mu-1} + \partial_x^{2\mu-1}s^\mu \right) + \partial_x^{-1}\varrho\partial_x^{-1}\varrho,
\end{align}
where $(\varrho)^2 = s^4$. As expected, $L\in \mathcal{D}^-$. The operator $L$ satisfies the additional condition that $L^*\partial_x = \partial_x L$, where the formal adjoint of an operator is defined by $(a\partial_x^k)^* = (-\partial_x)^k a$. In \cite{LWZ10}, the authors give the following proposition.

\begin{proposition}
There exists a unique pseudo-differential operator $P\in \mathcal{D}^-$, called the 6-th root of $L$, of the form
\begin{align*}
P = \partial_x + \sum_{k<0} p^k\partial_x^k
\end{align*}
such that $P^{6} = L$. The operator $P$ satisfies $[P,L] = 0$ and
\begin{align}\label{eqAdjointCondition6thRoot}
P^*\partial_x + \partial_x P = 0.
\end{align}
Moreover, Equation \eqref{eqAdjointCondition6thRoot} is equivalent to the condition that for every $k\in \mathbb{Z}^{\mathrm{odd}}_+$, the free term of the operator $(P^k)_+$ vanishes, i.e. $(P^k)_+(1) = 0$.
\end{proposition}

\noindent The above proposition implies that the following equations are well defined:
\begin{align}\label{eqScalarPositiveFlows}
\frac{\partial L}{\partial t_k} = \left[ \left( P^k\right)_+, L\right], \hspace{1cm} k\in \mathbb{Z}^\mathrm{odd}_+
\end{align}
These equation first appeared in \cite{DS84}, they coincide with Equations \eqref{eqMatrixLaxEqDn_positive}. Eventually, they give evolutionary PDEs on the functions $s^1,\ldots, s^4$.\\

\noindent \textbf{Scalar Lax pairs of the negative flows.} We now describe the scalar Lax pairs for the negative flows \eqref{eqMatrixLaxEqDn_negative}.
As mentioned above, we need to introduce the pseudo-differential operators of the second type, following \cite{LWZ10}. First of all, we extend the additive group $\mathcal{D}^-$ of pseudo-differential operators of the first type into the additive group
\begin{align*}
\mathcal{D} = \left\{\left. \sum_{k\in \mathbb{Z}} a^k\partial_x^k  \;\right|\; a^k\in \hcA^{[0]} \right\}.
\end{align*}
We use on $\hcA^{[0]}$ the gradation induced by $\deg s^\alpha_k = k$ and denote $\hcA^{[0]} = \bigoplus_{k\geq 0} \hcA^{[0]}_k$. We say that an operator $X\in \mathcal{D}^-\subset \mathcal{D}$ is \textit{homogeneous} of degree $k\in \mathbb{Z}$ if
\begin{align*}
X = \sum_{\ell\leq k} a^\ell \partial_x^{\ell}, \hspace{1cm} a^\ell\in \hcA^{[0]}_{k-\ell}.
\end{align*}
We denote by $\mathcal{D}_{k}$ the additive subgroup of $\mathcal{D}^-$ of homogeneous operators of the first type of degree $k$. Then $\mathcal{D} = \prod_{k\in\mathbb{Z}} \mathcal{D}_{k}.$ Similarly, for any $d\in\mathbb{Z}$ we define the additive subgroup $\mathcal{D}^+_{(d)} = \prod_{k\geq d} \mathcal{D}_{k}$. Now we define the additive subgroup
\begin{align*}
\mathcal{D}^+ = \bigcup_{d\in\mathbb{Z}} \mathcal{D}^+_{(d)}.
\end{align*}
The elements of $\mathcal{D}^+$ are the so-called \textit{pseudo-differential operators of the second type.} For these operators, we allow infinitely many nonzero terms of negative powers, as well as infinitely many nonzero terms of positive powers of $\partial_x$. In the end, the spaces $\mathcal{D}^-$ and $\mathcal{D}^+$ can be described as
\begin{align*}
\mathcal{D}^- &= \left\{\left. \sum_{k \leq m}\sum_{i\geq 0} a^{k,i}\partial_x^k \; \right| \; a^{k,i}\in \hcA^{[0]}_i \right\} & \mathcal{D}^+ &= \left\{\left. \sum_{k \geq m}\sum_{i\geq m-k} a^{k,i}\partial_x^k \; \right| \; a^{k,i}\in \hcA^{[0]}_i \right\}
\end{align*}
Both are complete topological rings. Also, we define the subgroup
\begin{align*}
\mathcal{D}^\mathrm{b} = \mathcal{D}^- \cap \mathcal{D}^+,
\end{align*}
whose elements are called \textit{bounded operators}. Since for any $X\in \mathcal{D}_{k}$, $Y\in \mathcal{D}_{\ell}$, their product satisfies $X\circ Y\in \mathcal{D}_{k+\ell}$, it follows that the product \eqref{eqProductOfOperators} in $\mathcal{D}^-$ can be extended to $\mathcal{D}^+$ to form a ring. The following proposition will be important for our computations \cite{LWZ10}.

\begin{proposition}
There exists a unique operator $Q\in \mathcal{D}^+$, caleld the square root of $L$, of the form
\begin{align}\label{eqSquareRootOfL}
Q = \partial_x^{-1}\varrho + \sum_{k\geq 0} Q_m\circ \partial_x
\end{align}
such that $Q^2 = L$. Here, $Q_m\in \Dcl^\mathrm{b}$ , is homogeneous of degree $2m$ and satisfies $Q_m^* = Q_m$. Moreover, the operator $Q$ satisfies
\begin{align}\label{eqSquareRootOfLProp}
Q^*\partial_x + \partial_x Q = 0, & & Q^*_+(\varrho) = - \sum_{m\geq 0} \partial_x Q_m(\varrho) = -\frac{1}{2}\partial_x L_+(1).
\end{align}
\end{proposition}

\noindent Notice that in the above proposition, $Q^*_+(\varrho)$ denotes the evaluation of the differential operator $Q^*_+$ at the function $\varrho$.  Thanks to this proposition, the following equations are well defined (we rewrite Equation \eqref{eqScalarPositiveFlows} for reasons of clarity):
\begin{equation}\label{eqScalarLaxDn}
\frac{\partial L}{\partial t_k} = \left[ \left( P^k\right)_+, L\right], \hspace{1cm} \frac{\partial L}{\partial \hat{t}_k} = \left[ \left( Q^k\right)_+, L\right], \hspace{1cm} k\in \mathbb{Z}^\mathrm{odd}_+.
\end{equation}
Finally, we can state the following theorem (\cite{LWZ10}, Theorem 4.11), which gives the complete picture of the scalar Lax pairs representation of the Drinfeld--Sokolov hierarchy of type $D_4$.

\begin{theorem}
The flows \eqref{eqMatrixLaxEqDn_positive}, \eqref{eqMatrixLaxEqDn_negative} of the Drinfeld--Sokolov  hierarchy of type $D_4$ coincide with the flows of Equation \eqref{eqScalarLaxDn}
\end{theorem}

\subsection{Bi-Hamiltonian structure} Here we give the two compatible Poisson brackets and Hamiltonian densities for the first Hamiltonian structure of the Drinfeld--Sokolov hierarchy of type $D_4$. They were given in \cite{DS84}, Proposition 8.3. To do so, we need to introduce the operator $\tilde{L} = \partial_x\circ L$ (denoted $\mathcal{L}$ in \cite{LWZ10}), which we write down using coordinates $v^\alpha$, namely,
\begin{align}
\tilde{L} &= \partial_x^{7} + \sum_{\mu = 1}^{3} \left(s^\mu\partial_x^{2\mu-1} + \partial_x^{2\mu-1} s^\mu \right) + \varrho\partial_x^{-1} \varrho \nonumber\\
	&= \partial_x^{7} + \sum_{\mu = 1}^{3} \left( v^\mu\partial_x^{2\mu -1} + \tilde{v}^\mu\partial_x^{2\mu-2} \right) + \varrho\partial_x^{-1} \varrho,\label{eqCoordinatesV}
\end{align}
and we set $v^4 = s^4 = (\varrho)^2$. The coordinates $\tilde{v}^\alpha$ are related to the coordinates $v^\alpha$ via the condition $\tilde{L}^* + \tilde{L} = 0$ (see Equation \eqref{eqVTildeToV}). In turns, the coordinates $v^\alpha$ are related to the coordinates $s^\alpha$ via a Miura transform by identifying the two expressions of $\tilde{L}$.

We call \textit{local functionals}, the elements of the quotient space 
$$\hLambda = \hcA/ \left( \mathrm{Im}(\partial_x)\oplus\mbC[[\eps]] \right).$$
Given a differential polynomial $f\in \hcA $, we denote by $\overline{f} = \int f dx$ its class in $\hLambda$. For a local functional $\overline{f}\in \hLambda$ in formal variables $v^*_*$, we can define its variational derivatives by
\begin{align*}
\frac{\delta \overline{f}}{\delta v^\alpha} = \sum_{k\geq 0} (-1)^k \partial_x^k\left(  \frac{\partial f}{\partial v^\alpha_k}\right) \in \hcA.
\end{align*}
It is well known that if $f\in \hcA$ is such that $f(0) = 0$ (no constant term), then $f\in \mathrm{Im}(\partial_x)$ if and only if $\smash{\frac{\delta \overline{f}}{\delta v^\alpha} = 0}$ for any $\alpha\in\{1,2,3,4\}$ (see e.g. \cite{GKMZ70}, Lemma 2). In particular, the variational derivatives are well defined on $\hLambda$. Next we define the variational differential (or variational derivative w.r.t. $\tilde{L}$) by
\begin{align*}
\frac{\delta \overline{f}}{\delta \tilde{L}} = \frac{\delta \overline{f}}{\delta v^4} + \frac{1}{2} \sum_{\mu = 1}^{3} \left( \frac{\delta \overline{f}}{\delta v^\mu} \partial_x^{-2\mu} + \partial_x^{-2\mu}\frac{\delta \overline{f}}{\delta v^\mu} \right) \in \Dcl^-.
\end{align*}
In \cite{DS84}, Drinfeld and Sokolov gave the following Poisson brackets: Let two local functionals $\overline{f},\overline{g}\in \hLambda$ and their variational differentials $X = \frac{\delta \overline{f}}{\delta \tilde{L}}, Y = \frac{\delta \overline{g}}{\delta \tilde{L}}$, then
\begin{align*}
\{ \overline{f},\overline{g}\}_1 &= \eps^{-1} \int \mathrm{res} \, X \left[ (\partial_x Y_+\tilde{L})_- - (\tilde{L} Y_+\partial_x)_- - (\partial_x Y_-\tilde{L})_+ + (\tilde{L} Y_-\partial_x)_+ \right] \ dx,\\
\{ \overline{f},\overline{g}\}_2 &= \eps^{-1}\int \mathrm{res} \, X \left[ (\tilde{L} Y)_+\tilde{L} - \tilde{L} (Y\tilde{L})_+ \right]\ dx.
\end{align*}
These brackets are compatible in the sense that for any $\lambda,\mu \in \mathbb{C}$, the map $\lambda\{\cdot,\cdot\}_1 + \mu\{\cdot,\cdot\}_2$ still satisfies Jacobi's identity. In \cite{LWZ10}, the authors proved the following theorem.

\begin{theorem}
The hierarchy \eqref{eqScalarLaxDn} admits the following bi-Hamiltonian representation: for any local functional $\overline{f}\in \hLambda$,
\begin{align}\label{eqHamD4}
\frac{\partial \overline{f}}{\partial t_k} &= \{\overline{f}, H_{k+6}\}_1 = \{\overline{f}, H_k\}_2, &
\frac{\partial \overline{f}}{\partial \hat{t}_k} &= \{\overline{f}, \hat{H}_{k+2}\}_1 = \{\overline{f}, \hat{H}_k\}_2,
\end{align}
where the Hamiltonian functionals are given by
\begin{align*}
H_k = \frac{6}{k}\int \mathrm{res}\, P^k dx, \hspace{1cm} \hat{H}_k = \frac{2}{k} \int \mathrm{res}\, Q^k dx
\end{align*}
\end{theorem}

	\subsection{Tau structure}

Finally, we present the tau structure (in the sense of \cite{DZ05}) of the Drinfeld--Sokolov hierarchy of type $D_4$. The first step is to define the so-called topological variables: for $\mu\in \{1, 2, 3\}$ and $p\geq 0$,
\begin{align*}
t^\mu_p = \frac{6\Gamma\left(p+1+\frac{2\mu-1}{6}\right)}{\Gamma\left( \frac{2\mu-1}{6}\right)} t_{6p+2\mu-1}, \hspace{1cm} t^4_p = \frac{2\Gamma\left( p+\frac{3}{2} \right)}{\Gamma(\frac{1}{2})} \hat{t}_{2p+1}.
\end{align*}
In particular, $t^1_0 = t_1 = x$. In the same fashion, we define the following Hamiltonian densities:
\begin{align}\label{eqTauSymHamDens}
h_{\mu,p-1} &= \frac{\Gamma\left( \frac{2\mu-1}{6} \right)}{6 \Gamma\left( p+1+\frac{2\mu-1}{6} \right)} \mathrm{res}\, P^{6p +2\mu-1}, & h_{4,p-1} &= \frac{\Gamma(\frac{1}{2})}{2\Gamma\left( p+\frac{3}{2} \right)} \mathrm{res}\, Q^{2p+1},\\
 &= 6^p\prod_{j = 0}^{p} (2\mu-1 + 6j)^{-1}\, \mathrm{res}\, P^{6p +2\mu-1}, & &= 2^p\prod_{j=0}^p(1+2j)^{-1}\, \mathrm{res}\, Q^{2p+1}. \nonumber
\end{align}
We denote by $\overline{h}_{\alpha,p} = \int h_{\alpha,p} dx$ the associated functionals. Then the Hamiltonian equations \eqref{eqHamD4} read
\begin{align*}
\frac{\partial \overline{f}}{\partial t^\alpha_p} = \{ \overline{f}, \oh_{\alpha,p}\}_1 = \left( p+\textstyle\frac{1}{2} + \mu_\alpha \right)^{-1} \{\overline{f}, \oh_{\alpha,p-1} \}_2,
\end{align*}
where $\mu_\alpha$'s are the spectrum of the underlying Frobenius manifold \cite{Dub96,DLZ08}; they read $\mu_\nu = \frac{2\nu - 4}{6}$, for $\nu\in \{1,2,3\}$, and $\mu_4 = 0$. These Hamiltonian densities satisfy the so-called tau symmetry \cite{DZ05}, \cite{LWZ10}: for any $\alpha,\beta\in \{1,2,3,4\}$ and $p,q\geq 1$,
\begin{align}
\frac{\partial h_{\alpha,p-1}}{\partial t^\beta_q} = \frac{\partial h_{\beta,q-1}}{\partial t^\alpha_p}.
\end{align}
If we define differential polynomials $\Omega_{\alpha,p; \beta,q} $ via $\partial_x\,\Omega_{\alpha,p; \beta,q} = \partial_{t^\beta_q} h_{\alpha,p-1} $, then we can find integration constants such that they satisfy $\Omega_{\alpha,p;\beta,q} = \Omega_{\beta,q;\alpha,p}$. This is a sufficient condition to define the tau function $\tau$ of the hierarchy by setting \cite{DZ05}
\begin{align*}
\frac{\partial^2 \log\tau}{\partial t^\alpha_p \partial t^\beta_q} = \Omega_{\alpha,p;\beta,q}.
\end{align*}

	\subsection{Explicit tau-symmetric Hamiltonian densities}

In this section, we give the explicit fomulae for the Hamiltonian densities $h_{\alpha,-1}$, $h_{\alpha,0}$, for $1\leq \alpha \leq 4$, and $h_{1,1}$. As proven in \cite{BG16}, even $\overline{h}_{1,1}$ (and the Poisson structure) suffices to confirm the equivalence with the double ramification hierarchy (see Section \ref{section:DRclassical}). Our computations will be expressed in a special set of coordinates called \textit{normal} and defined by
\begin{align*}
\tu^\alpha = \eta^{\alpha\beta}h_{\beta,-1}
\end{align*}
(summation over $1\leq \beta\leq 4$ is implicit) where the matrix $(\eta^{\alpha\beta})$ is given by (see e.g. \cite{LRZ15})
\begin{align*}
(\eta^{\alpha\beta}) = \begin{pmatrix}
0 & 0 & 6 & 0\\
0 & 6 & 0 & 0\\
6 & 0 & 0 & 0\\
0 & 0 & 0 & 2
\end{pmatrix}
\end{align*}
By Equation \eqref{eqTauSymHamDens}, the normal coordinates read, for $\mu\in\{1,2,3\}$,
\begin{align*}
\tu^\mu = \frac{6}{7-2\mu} \res\, P^{7-2\mu}, \hspace{1cm} \tu^4 = 2\,\res\, Q.
\end{align*}
Because of the form of the operator $Q$ \eqref{eqSquareRootOfL}, its residue is straightforwardly $\res\,Q = \varrho$, which means that $\tu^4 = 2\varrho = 2\sqrt{s^4}$. For the remaining normal coordinates, we compute the residues in the coordinates $s^\alpha$ and then invert the system. We find
\begin{align}\label{eqMiuraTransDStoNormal}
\left\{ \begin{aligned}
s^1 &= \ts \frac{1}{2}\tu^1 + \frac{1}{12}\tu^2\tu^3 + \frac{1}{216}\left(  \tu^3 \right) ^3 + \left( - \frac{1}{8}\left(  \tu^3_1 \right) ^2 - \frac{1}{6}\tu^2_2 - \frac{1}{9}\tu^3\tu^3_2 \right)\eps^2 + \frac{23}{90}\tu^3_4\eps^4,\\
s^2 &= \ts \frac{1}{2}\tu^2 + \frac{1}{8}\left(  \tu^3 \right) ^2 - \frac{1}{2} \tu^3_2\eps^2,\\
s^3 &= \ts \frac{1}{2}\tu^3,\\
s^4 &= \ts \frac{1}{4}\left( \tu^4 \right)^2.
\end{aligned} \right.
\end{align}
(We have performed the substitution $\partial_x^k(f) \to \left( \eps /\sqrt{2} \right)^k \partial_x^k(f)$.) This agrees with the expressions found in \cite{LRZ15}, p. 751. Then the Hamiltonian densities $h_{\alpha,0}$ are given by, for $\mu\in\{1,2,3\}$,
\begin{align*}
h_{\mu,0} = \frac{6}{(2\mu-1)(2\mu+5)} \res\, P^{2\mu+5}, \hspace{1cm} h_{4,0} = \frac{2}{3}\res\, Q^3
\end{align*}
For the computation $\res\, Q^3$ we write
\begin{align*}
\res\, Q^3 = \res\, QL &= \res\left( \partial_x^{-1}\varrho L \right) + \res\left( \sum_{m\geq 0}Q_m \partial_x L \right),\\
 &= \res\left( \partial_x^{-1}\varrho L \right) + \res\left( {\sum_{m\geq 0}} Q_m\varrho\partial_x^{-1}\varrho \right),
\end{align*}
where in the last equation we used the fact that $\partial_x L = \tilde{L} = \tilde{L}_+ + \varrho\partial_x^{-1}\varrho$. To compute the rightmost term, we write $Q_m = \sum_{k\leq 2m} q^{m,k}\partial_x^k$, then
\begin{align*}
\res\left( \sum_{m\geq 0} Q_m\varrho\partial_x^{-1}\varrho \right) &= \res\left( \sum_{m\geq 0}\sum_{k\leq 2m} q^{m,k}\partial_x^k \varrho\partial_x^{-1}\varrho \right)\\
 &= \res\left( \sum_{m\geq 0} \sum_{k\leq 2m} q^{m,k} \sum_{\ell = 0}^{k} \binom{k}{\ell} \left(\frac{\eps}{\sqrt{2}}\right)^\ell \varrho_{\ell} \partial_x^{k-1 -\ell} \varrho \right)\\
 &= \sum_{m\geq 0} \sum_{k\leq 2m} q^{m,k} \left(\frac{\eps}{\sqrt{2}}\right)^k \varrho_{k}\varrho = \varrho \sum_{m\geq 0} Q_m(\varrho).
\end{align*}
In the above equations, all matters of convergence are resolved by the grading of $\hcA$ and the fact that $Q\in \Dcl^+$. Now thanks to Equation \eqref{eqSquareRootOfLProp}, it follows that
\begin{align*}
\sum_{m\geq 0} Q_m(\varrho) = \frac{1}{2}L_+(1) = \sum_{\mu = 1}^3 s^\mu_{2\mu - 2} \frac{\eps^{2\mu-2}}{2^{\mu-1}}.
\end{align*}
Using $\tu^4 = 2\varrho$, we finally find that
\begin{align*}
h_{4,0} = \frac{1}{3} \res\left(\partial_x^{-1}\tu^4 L\right) + \frac{1}{3}\tu^4\sum_{\mu = 1}^3 s^\mu_{2\mu -2}\frac{\eps^{2\mu-2}}{2^{\mu-1}}.
\end{align*}
 We find:
\begin{align*}
\begin{split}
h_{1,0}\ =\ & \ts
\ts \left(\frac{1}{12}\left(\tu^2\right)^2+\frac{1}{6}\tu^1 \tu^3+\frac{1}{4}\left(\tu^4\right)^2\right)+\left(\frac{1}{72} \tu^3 \left(\tu_1^3\right)^2+\frac{1}{3}\tu_2^1+\frac{1}{216} \left(\tu^3\right)^2 \tu_2^3\right) \eps
   \ts ^2+\left(\frac{1}{216} \left(\tu_2^3\right)^2+\frac{1}{72} \tu_1^3 \tu_3^3\right. \\
   & \ts \left. +\frac{1}{180} \tu^3 \tu_4^3\right) \eps ^4+\frac{1}{840} \tu_6^3 \eps ^6\\
   h_{2,0}\ =\ & \ts
\ts \left(\frac{1}{6}\tu^1 \tu^2-\frac{1}{72} \left(\tu^2\right)^2 \tu^3+\frac{1}{648} \tu^2 \left(\tu^3\right)^3+\frac{1}{24} \tu^3 \left(\tu^4\right)^2\right)+\left(-\frac{1}{24} \left(\tu_1^2\right)^2+\frac{1}{24} \tu^3 \tu_1^2
   \ts \tu_1^3+\frac{1}{72} \tu^2 \left(\tu_1^3\right)^2 \right. \\
   & \ts \left. + \frac{1}{8} \left(\tu_1^4\right)^2-\frac{1}{18} \tu^2 \tu_2^2+\frac{1}{72} \left(\tu^3\right)^2 \tu_2^2+\frac{1}{54} \tu^2 \tu^3 \tu_2^3+\frac{1}{6} \tu^4 \tu_2^4\right)
   \ts \eps ^2+\left(\frac{1}{12} \tu_2^2 \tu_2^3+\frac{5}{72} \tu_1^3 \tu_3^2+\frac{1}{18} \tu_1^2 \tu_3^3 \right. \\
   & \ts \left. +\frac{1}{36} \tu^3 \tu_4^2+\frac{2}{135} \tu^2 \tu_4^3\right) \eps ^4+\frac{1}{72} \tu_6^2 \eps ^6
   \end{split}
\end{align*}
\begin{align*}
\begin{split}
\ts h_{3,0}\ =\ &
\ts \left(\frac{1}{12}\left(\tu^1\right)^2-\frac{1}{216}\left(\tu^2\right)^3+\frac{1}{432} \left(\tu^2\right)^2 \left(\tu^3\right)^2+\frac{1}{233280}\left(\tu^3\right)^6+\frac{1}{24} \tu^2 \left(\tu^4\right)^2+\frac{1}{144}
   \ts \left(\tu^3\right)^2 \left(\tu^4\right)^2\right)+ \\
   & \ts \left(\frac{5}{432} \tu^3 \left(\tu_1^2\right)^2+\frac{1}{108} \tu^3 \tu_1^1 \tu_1^3+\frac{1}{36} \tu^2 \tu_1^2 \tu_1^3+\frac{1}{1296}\left(\tu^3\right)^3
   \ts \left(\tu_1^3\right)^2+\frac{1}{12} \tu^4 \tu_1^3 \tu_1^4+\frac{5}{144} \tu^3 \left(\tu_1^4\right)^2 \right. \\
   & \ts \left.+\frac{1}{216} \left(\tu^3\right)^2 \tu_2^1+\frac{1}{54} \tu^2 \tu^3 \tu_2^2+\frac{1}{108} \left(\tu^2\right)^2
   \ts \tu_2^3+\frac{1}{3888}\left(\tu^3\right)^4 \tu_2^3+\frac{1}{36} \left(\tu^4\right)^2 \tu_2^3+\frac{1}{18} \tu^3 \tu^4 \tu_2^4\right) \eps ^2 \\
   & \ts +\left(\frac{5}{2592}\left(\tu_1^3\right)^4+\frac{13}{720}
   \ts \left(\tu_2^2\right)^2+\frac{55}{3888} \tu^3 \left(\tu_1^3\right)^2 \tu_2^3+\frac{29}{1080} \tu_2^1 \tu_2^3+\frac{67}{15552} \left(\tu^3\right)^2 \left(\tu_2^3\right)^2+\frac{13}{240}
   \ts \left(\tu_2^4\right)^2 \right. \\
   & \ts \left.+\frac{1}{30} \tu_1^3 \tu_3^1+\frac{31}{1080} \tu_1^2 \tu_3^2+\frac{11}{1080} \tu_1^1 \tu_3^3+\frac{49}{7776} \left(\tu^3\right)^2 \tu_1^3 \tu_3^3+\frac{31}{360} \tu_1^4 \tu_3^4+\frac{1}{60} \tu^3 \right.
   \ts \tu_4^1+\frac{2}{135} \tu^2 \tu_4^2 \\
   & \ts \left. +\frac{1}{1296}\left(\tu^3\right)^3 \tu_4^3+\frac{2}{45} \tu^4 \tu_4^4\right) \eps ^4+\left(\frac{1129}{116640} \left(\tu_2^3\right)^3+\frac{1601}{38880} \tu_1^3 \tu_2^3\tu_3^3+\frac{1}{120} \tu^3 \left(\tu_3^3\right)^2+\frac{19}{1620} \left(\tu_1^3\right)^2 \tu_4^3 \right. \\
   & \ts \left. +\frac{29}{2160} \tu^3 \tu_2^3 \tu_4^3+\frac{13}{1944} \tu^3 \tu_1^3 \tu_5^3+\frac{11}{1080} \tu_6^1 + \frac{17}{19440}\left(\tu^3\right)^2 \tu_6^3\right) \eps ^6+\left(\frac{191}{43200} \left(\tu_4^3\right)^2+\frac{1501}{194400} \tu_3^3 \tu_5^3\right.\\
   & \ts \left. +\frac{949}{194400} \tu_2^3 \tu_6^3+\frac{127}{64800} \tu_1^3 \tu_7^3+\frac{13}{32400} \tu^3 \tu_8^3\right) \eps ^8+\frac{7}{118800} \tu_{10}^3 \eps ^{10}\\
h_{4,0}\ =\ & \ts 
\ts \left(\frac{1}{2}\tu^1 \tu^4+\frac{1}{12} \tu^2 \tu^3 \tu^4+\frac{1}{216} \left(\tu^3\right)^3 \tu^4\right)+\left(\frac{1}{24} \tu^4 \left(\tu_1^3\right)^2+\frac{1}{4} \tu_1^2 \tu_1^4+\frac{1}{8} \tu^3 \tu_1^3 \tu_1^4+\frac{1}{6} \tu^4
   \ts \tu_2^2 \right. \\
   & \ts \left. +\frac{1}{18} \tu^3 \tu^4 \tu_2^3+\frac{1}{6} \tu^2 \tu_2^4+\frac{1}{24} \left(\tu^3\right)^2 \tu_2^4\right) \eps ^2+\left(\frac{1}{4} \tu_2^3 \tu_2^4+\frac{1}{6} \tu_1^4 \tu_3^3+\frac{5}{24} \tu_1^3 \tu_3^4+\frac{2}{45}
   \ts \tu^4 \tu_4^3+\frac{1}{12} \tu^3 \tu_4^4\right) \eps ^4  \\
   &  \ts +\frac{1}{24} \tu_6^4 \eps ^6
\end{split}
\end{align*}
Finally, the density $h_{1,1}$ is given by
\begin{align*}
h_{1,1} = \frac{36}{91} \res\, P^{13}.
\end{align*}
We find:
\begin{align}\label{eq:DS11}
\begin{split}
h_{1,1} =\ & \textstyle \left( \frac{1}{12} \tu^1 \left(\tu^2\right)^2+\frac{1}{12} \left(\tu^1\right)^2u^3-\frac{1}{108} \left(\tu^2\right)^3 \tu^3+\frac{1}{432} \left(\tu^2\right)^2 \left(\tu^3\right)^3+\frac{1}{326592}\left(\tu^3\right)^7+\frac{1}{4}\tu^1 \left(\tu^4\right)^2 \right.\\
	& \ts \left. +\frac{1}{12} \tu^2 \tu^3 \left(\tu^4\right)^2+\frac{1}{144} \left(\tu^3\right)^3 \left(\tu^4\right)^2 \right) + \left( \frac{1}{6} \left(\tu^1_1\right)^2+\frac{1}{3} \tu^1_2u^1-\frac{5}{72} \left(\tu^2_1\right)^2 \tu^2+\frac{5}{24} \left(\tu^4_1\right)^2 \tu^2-\frac{1}{18}\tu^2_2 \left(\tu^2\right)^2 \right.\\
	& \ts \left. +\frac{1}{54} \left(\tu^3_1\right)^2\left(\tu^2\right)^2+\frac{1}{72} \left(\tu^3_1\right)^2 \tu^1u^3+\frac{19}{216} \tu^2_1u^3_1 \tu^2 \tu^3+\frac{1}{36} \tu^3_2 \left(\tu^2\right)^2u^3+\frac{1}{54} \left(\tu^2_1\right)^2 \left(\tu^3\right)^2+\frac{1}{72} \tu^1_1 \tu^3_1 \left(\tu^3\right)^2 \right.\\
	& \ts \left. +\frac{1}{18}\left(\tu^4_1\right)^2 \left(\tu^3\right)^2+\frac{1}{216} \tu^3_2 \tu^1\left(\tu^3\right)^2+\frac{7}{216} \tu^2_2 \tu^2 \left(\tu^3\right)^2+\frac{1}{216} \tu^1_2\left(\tu^3\right)^3+\frac{7}{7776} \left(\tu^3_1\right)^2 \left(\tu^3\right)^4 \right.\\
	& \ts \left. +\frac{1}{3888}\tu^3_2 \left(\tu^3\right)^5+\frac{5}{12}\tu^2_1 \tu^4_1 \tu^4+\frac{1}{3}\tu^4_2 \tu^2 \tu^4+\frac{19}{72} \tu^3_1 \tu^4_1 \tu^3 \tu^4+\frac{7}{72}\tu^4_2 \left(\tu^3\right)^2 \tu^4+\frac{1}{6} \tu^2_2 \left(\tu^4\right)^2 \right.\\
	& \ts \left. +\frac{1}{18}\left(\tu^3_1\right)^2 \left(\tu^4\right)^2+\frac{1}{12} \tu^3_2 \tu^3\left(\tu^4\right)^2 \right)\eps^2 + \left( \frac{7}{24} \tu^2_1 \tu^2_2 \tu^3_1+\frac{5}{72} \tu^1_2\left(\tu^3_1\right)^2+\frac{31}{216} \left(\tu^2_1\right)^2 \tu^3_2+\frac{13}{216} \tu^1_1u^3_1 \tu^3_2 \right.\\
	& \ts \left. +\frac{31}{72} \tu^3_2\left(\tu^4_1\right)^2+\frac{7}{8} \tu^3_1 \tu^4_1 \tu^4_2+\frac{1}{216} \left(\tu^3_2\right)^2 \tu^1+\frac{1}{72}\tu^3_1 \tu^3_3 \tu^1+\frac{35}{216} \tu^2_3 \tu^3_1u^2+\frac{23}{108} \tu^2_2 \tu^3_2 \tu^2+\frac{4}{27} \tu^2_1 \tu^3_3 \tu^2 \right.\\
	& \ts \left. +\frac{7}{270} \tu^3_4\left(\tu^2\right)^2+\frac{13}{144} \left(\tu^2_2\right)^2u^3+\frac{29}{216} \tu^2_1 \tu^2_3 \tu^3+\frac{1}{9} \tu^1_3 \tu^3_1 \tu^3+\frac{7}{648} \left(\tu^3_1\right)^4u^3+\frac{19}{216} \tu^1_2 \tu^3_2 \tu^3+\frac{7}{216} \tu^1_1u^3_3 \tu^3 \right.\\
	& \ts \left. +\frac{13}{48} \left(\tu^4_2\right)^2 \tu^3+\frac{29}{72}\tu^4_1 \tu^4_3 \tu^3+\frac{1}{180}\tu^3_4 \tu^1 \tu^3+\frac{7}{108} \tu^2_4 \tu^2 \tu^3+\frac{1}{36} \tu^1_4\left(\tu^3\right)^2+\frac{53}{1296} \left(\tu^3_1\right)^2 \tu^3_2\left(\tu^3\right)^2 \right.\\
	& \ts \left. 
+\frac{133}{15552}\left(\tu^3_2\right)^2 \left(\tu^3\right)^3+\frac{97}{7776} \tu^3_1u^3_3 \left(\tu^3\right)^3+\frac{1}{864} \tu^3_4 \left(\tu^3\right)^4+\frac{4}{9}\tu^3_3 \tu^4_1 \tu^4+\frac{23}{36} \tu^3_2 \tu^4_2u^4+\frac{35}{72} \tu^3_1 \tu^4_3 \tu^4 \right.\\
	& \ts \left. +\frac{7}{36} \tu^4_4 \tu^3u^4+\frac{7}{90} \tu^3_4 \left(\tu^4\right)^2 \right) \eps^4 + \left( \frac{349}{3024} \left(\tu^2_3\right)^2+\frac{4}{21} \tu^2_2 \tu^2_4+\frac{17}{168} \tu^2_1 \tu^2_5+\frac{65}{504} \tu^1_5 \tu^3_1+\frac{151}{756}\tu^1_4 \tu^3_2 \right.\\
	& \ts \left. +\frac{473}{2592} \left(\tu^3_1\right)^2 \left(\tu^3_2\right)^2+\frac{31}{168} \tu^1_3 \tu^3_3+\frac{19}{216} \left(\tu^3_1\right)^3 \tu^3_3+\frac{19}{210}\tu^1_2 \tu^3_4+\frac{149}{7560} \tu^1_1 \tu^3_5+\frac{349}{1008}\left(\tu^4_3\right)^2+\frac{4}{7} \tu^4_2 \tu^4_4 \right.\\
	& \ts \left. +\frac{17}{56} \tu^4_1 \tu^4_5+\frac{1}{840}\tu^3_6 \tu^1+\frac{53}{1512} \tu^2_6 \tu^2+\frac{65}{1512} \tu^1_6u^3+\frac{1741}{23328} \left(\tu^3_2\right)^3 \tu^3+\frac{2507}{7776}\tu^3_1 \tu^3_2 \tu^3_3u^3+\frac{593}{6480} \left(\tu^3_1\right)^2 \tu^3_4u^3 \right.\\
	& \ts \left. +\frac{271}{7776} \left(\tu^3_3\right)^2 \left(\tu^3\right)^2+\frac{2141}{38880} \tu^3_2 \tu^3_4\left(\tu^3\right)^2+\frac{341}{12960} \tu^3_1 \tu^3_5 \left(\tu^3\right)^2+\frac{43}{19440} \tu^3_6 \left(\tu^3\right)^3+\frac{53}{504} \tu^4_6u^4 \right)\eps^6 + \left( \frac{1}{54}\tu^1_8 \right.\\
	& \ts \left. +\frac{653}{1944} \tu^3_2 \left(\tu^3_3\right)^2+\frac{17803}{68040} \left(\tu^3_2\right)^2 \tu^3_4+\frac{11129}{30240}\tu^3_1 \tu^3_3 \tu^3_4+\frac{65141}{272160}\tu^3_1 \tu^3_2 \tu^3_5+\frac{145}{3024} \left(\tu^3_1\right)^2u^3_6+\frac{17503}{302400} \left(\tu^3_4\right)^2u^3 \right.\\
	& \ts \left. +\frac{977}{10080} \tu^3_3 \tu^3_5u^3+\frac{15103}{272160} \tu^3_2 \tu^3_6 \tu^3+\frac{1831}{90720}\tu^3_1 \tu^3_7 \tu^3+\frac{19}{9720} \tu^3_8 \left(\tu^3\right)^2 \right)\eps^8 + \left( \frac{9973}{340200} \left(\tu^3_5\right)^2+\frac{1301}{25200} \tu^3_4 \tu^3_6 \right.\\
	& \ts \left. +\frac{347}{10080} \tu^3_3 \tu^3_7+\frac{4427}{272160} \tu^3_2u^3_8+\frac{89}{18144}\tu^3_1 \tu^3_9+\frac{1}{1296}\tu^3_{10} \tu^3 \right)\eps^{10} + \frac{41}{393120} \tu^3_{12}\eps^{12}
\end{split}
\end{align}

\subsection{Explicit first Poisson structure}

Finally, we compute the components of the first Poisson structure given by, for two local functionals $\overline{f},\overline{g}\in \hLambda$ and their variational differentials $X = \frac{\delta \overline{f}}{\delta \tilde{L}}, Y = \frac{\delta \overline{g}}{\delta \tilde{L}}$,
\begin{align}\label{eqFirstPoissonBracket}
\{ \overline{f},\overline{g}\}_1 &= \frac{1}{\eps}\int \mathrm{res} \, X \left[ (\partial_x Y_+\tilde{L})_- - (\tilde{L} Y_+\partial_x)_- - (\partial_x Y_-\tilde{L})_+ + (\tilde{L} Y_-\partial_x)_+ \right] dx.
\end{align}
Working with normal coordinate $\tu^*_*$, we aim to express this bracket through an Hamiltonian operator $K_\tu^{\alpha\beta} = \sum_{k\geq 0} K^{\alpha\beta}_k \d_x^k$, where $ K^{\alpha\beta}_k \in \hcA_\tu^{[-k+1]}$, via
\begin{align*}
\{ \overline{f}, \overline{g} \}_1 = \int \frac{\delta \overline{f}}{\delta \tu^\alpha} K_\tu^{\alpha\beta} \frac{\delta \overline{g}}{\delta \tu^\beta}.
\end{align*}
Under a Miura transformation $\tu^*_* = \tu^*_*(u^*_*,\eps)$, the Hamiltonian operator transforms accoring to $K_{\widetilde u}^{\alpha\beta} = (L^*)^\alpha_\mu \circ K_u^{\mu\nu} \circ L^\beta_\nu$,
where $(L^*)^\alpha_\mu = \sum\limits_{s\geq 0} \frac{\partial \widetilde u^\alpha}{\partial u^\mu_s} \partial_x^s$, $L^\beta_\nu = \sum\limits_{s\geq 0} (-\partial_x)^s \circ\frac{\partial \widetilde u^\beta}{\partial u^\nu_s}$. When clear from the context which formal variables are being used, as usual, we will drop the corresponding index in $K^{\alpha\beta}_u$ or $K^{\alpha\beta}_\tu$ and simply write $K^{\alpha\beta}$.
We use Dirac's notation of the components $K^{\alpha\beta}$ of the bracket and rewrite the above expression in the equivalent form
\begin{align}
\{ \overline{f},\overline{g} \}_1 = \iint \frac{\delta \overline{f}(x)}{\delta \tu^\alpha(x)} \{\tu^\alpha(x),\tu^\beta(y)\}_1 \frac{\delta \overline{g}(y)}{\delta \tu^\beta(y)} \mathrm{d}x\mathrm{d}y,
\end{align}
where
\begin{align}\label{eqPoissonCompGen}
\{\tu^\alpha(x),\tu^\beta(y)\}_1 := \sum_{k \geq 0} K^{\alpha\beta}_k(x) \delta^{(k)}(x-y).
\end{align}
In Equation \eqref{eqPoissonCompGen}, the expression ``$\{\tu^\alpha(x),\tu^\beta(y)\}_1$'' really is only a notation and does not mean that we evaluate the bracket on the pair of functions $(\tu^\alpha(x),\tu^\beta(y))$. However, one can evaluate the bracket on the pair of functions $(\tu^\alpha(z)\delta(x-z),\tu^\beta(z)\delta(y-z))$ (sometimes called the \textit{coordinate functionals} in physics literature) and find that, almost tautologically,
\begin{align*}
\{\tu^\alpha(z)\delta(x-z), \tu^\beta\delta(y-z)\}_1 &= \iint \frac{\delta \tu^\alpha(z)\delta(x-z)}{\delta \tu^\mu(z)} \sum_{k\geq 0} K^{\mu\nu}_k \delta^{(k)}(z-z') \frac{\delta \tu^\alpha(z')\delta(x-z')}{\delta \tu^\mu(z')} \mathrm{d}z\mathrm{d}z'\\
 &= \sum_{k \geq 0} K^{\alpha\beta}_k(x) \delta^{(k)}(x-y).
\end{align*}
Hence,
\begin{align*}
\{\tu^\alpha(x),\tu^\beta(y)\}_1 = \{\tu^\alpha(z)\delta(x-z), \tu^\beta\delta(y-z)\}_1,
\end{align*}
where in the above equation, the left-hand side is purely notational, while the right-hand side can actually be computed.

Now the Poisson bracket \eqref{eqFirstPoissonBracket} is given in terms of the variational differentials, so that we first need to compute those of the coordinate functionals. To do so, we need to transform the normal coordinates $\tu^\alpha$ into the coordinates $v^\alpha$ of Equation \eqref{eqCoordinatesV}. The coordinates $\tilde{v}^\alpha$ are related to the $v^\alpha$'s by the condition $\tilde{L}^* + L = 0$; they read
\begin{align}\label{eqVTildeToV}
\left\{ \begin{aligned}
\tilde{v}^1 &= \ts \frac{1}{2\sqrt{2}}v^1_1\eps - \frac{1}{8\sqrt{2}}v^2_3\eps^3 + \frac{1}{8\sqrt{2}}v^3_5\eps^5,\\
\tilde{v}^2 &= \ts \frac{3}{2\sqrt{2}}v^2_1\eps - \frac{5}{4\sqrt{2}}v^3_3\eps^3,\\
\tilde{v}^3 &= \ts \frac{5}{2\sqrt{2}}v^3_1\eps.
\end{aligned} \right.
\end{align}
Then by identifying the two operators in Equation \eqref{eqFirstPoissonBracket} and inverting Equation \eqref{eqMiuraTransDStoNormal}, we find the following Miura transformation from the normal coordinates $\tu^\alpha$ to the coordinates $v^\alpha$:
\begin{align*}
\left\{ \begin{aligned}
\tu^1 &= \ts v^1 - \frac{1}{6} v^2v^3 + \frac{7}{216}(v^3)^3 + \left( \frac{1}{12}(v^3_1)^2 - \frac{5}{12}v^2_2 + \frac{11}{36}v^3v^3_2 \right)\eps^2 + \frac{89}{90}v^3_4\eps^4,\\
\tu^2 &= \ts v^2 - \frac{1}{4}(v^3)^2 - \frac{3}{2}v^3_2\eps^2,\\
\tu^3 &= \ts v^3,\\
\tu^4 &= \ts 2\sqrt{v^4}.
\end{aligned} \right.
\end{align*}
Now recall that the variational differentials are given by
\begin{align*}
\frac{\delta F(z)}{\delta \tilde{L}(z)} = \frac{\delta F(z)}{\delta v^4(z)} + \frac{1}{2} \sum_{\mu = 1}^3 \left( \frac{\delta F(z)}{\delta v^\mu(z)} \partial^{-2\mu}_z + \partial^{-2\mu}_z\frac{\delta F(z)}{\delta v^\mu(z)} \right) \in \Dcl^-.
\end{align*}
It follows that the variational differentials of the coordinates functionals are given by, for $\mu\in\{1,2,3\}$,
\begin{align}
\frac{\delta\, \tu^\mu(z)\delta(x-z)}{\delta \tilde{L}(z)} &= \frac{1}{2}\sum_{\nu=1}^3 \left( \frac{\delta \tu^\mu(z) \delta(x-z)}{\delta v^\nu(z)} \partial_z^{-2\nu} + \partial_z^{-2\nu} \frac{\delta \tu^\mu(z) \delta(x-z)}{\delta v^\nu(z)} \right),\label{eqVarDiffua}\\
\frac{\delta\, \tu^4(z)\delta(x-z)}{\delta \tilde{L}(z)} &= \frac{\delta \left( 2\sqrt{v^4(z)}\delta(x-z) \right)}{\delta v^4(z)} = \frac{2}{u^4(z)}\delta(x-z). \label{eqVarDiffu4}
\end{align}
Let us denote, for $\alpha\in\{1,2,3,4\}$, 
\begin{align*}
X^\alpha(z) := \frac{\delta\, \tu^\alpha(z)\delta({x}-z)}{\delta \tilde{L}(z)}, \hspace{1cm} Y^\alpha(z) := \frac{\delta\, \tu^\alpha(z)\delta({y}-z)}{\delta \tilde{L}(z)}
\end{align*}
As we can see in Equations \eqref{eqVarDiffua} and \eqref{eqVarDiffu4}, for any $\mu\in\{1,2,3\}$, we have $Y^\mu(z)_+ = 0$, and $Y^4(z)_- = 0$. It follows that
\begin{align*}
\{\tu^4(x),\tu^4(y)\}_1 =\ & 2\delta'(x-y)
\end{align*}
(Note that this computation does not depend on the integer $n$ of $D_n$, meaning that similarly, in the $D_n$ case, $\{\tu^n(x),\tu^n(y)\}_1 = 2\delta'(x-y)$.) It is easy to see that for any $\mu\in \{1,2,3\}$, $\{\tu^\mu(x),\tu^4(y)\}_1 = 0$. The remaining components are computed in a similar fashion, this time using
\begin{align*}
\{\tu^\mu(x),\tu^\nu(y)\}_1 =\ & \frac{1}{\eps} \displaystyle \int \ts \res\, X^\mu(z) \left[ \left( \tilde{L}(z) Y^\nu(z) \partial_z \right)_+ - \left( \partial_z Y^\nu(z) \tilde{L}(z) \right)_+ \right] \mathrm{d}z,
\end{align*}
for $\mu,\nu\in\{1,2,3\}$. Now we equivalently write $\partial_x^k$ instead of $\delta^{(k)}(x-y)$, for their action is identical. We write down the components $\{\tu^\alpha(x),\tu^\beta(y)\}_1$ in the following matrix:
\begin{equation}\label{eqPoissonMatrix}
K^{\mathrm{DS}}_{\tu} = \ts \begin{pmatrix}
\left(\frac{1}{6}\tu^3_2\partial_x + \frac{1}{2}\tu^3_1\partial_x^2 + \frac{1}{3}\tu^3\partial_x^3 \right)\eps^2 + \frac{4}{15}\partial_x^5\eps^4 & 0 & 6\partial_x & 0\\
0 & 6\partial_x & 0 & 0\\
6\partial_x & 0 & 0 & 0\\
0 & 0 & 0 & 2\partial_x
\end{pmatrix} = \eta\partial_x + \mathcal{O}(\eps)
\end{equation}

\section{Classical double ramification hierarchy for the $D_4$ Dubrovin-Saito CohFT} \label{section:DRclassical}

In this section we compute the double ramification hierarchy for the $D_4$ Dubrovin-Saito or Fan-Jarvis-Ruan-Witten cohomological field theory (the latter with respect to the non-maximal diagonal symmetry group $\<J\> = \mbZ/3\mbZ$) in the framework of homogeneous integrable systems of double ramification type. This means that we will find the unique dispersive deformation of double ramification type, according to \cite{BDGR16b}, of the principal hierarchy associated to the Frobenius manifold for the $D_4$ Coxeter group \cite{Dub99}, compatible with the natural grading of this homogeneous CohFT.\\

Next we will show that this $D_4$ DR hierarchy coincides with the $D_4$ Drinfeld-Sokolov hierarchy described in the previous section. Indeed, in \cite{GM05,FGM10} (see also \cite{LRZ15}) the authors proved that the $D_4$ Drinfeld-Sokolov hierarchy is the Dubrovin-Zhang hierarchy of the aforementioned $D_4$ CohFT and, as proven in \cite{BG16}, for a hierarchy of DR type to coincide with a DZ hierarchy with the same dispersionless (i.e. genus $0$) limit, it is enough that their Hamiltonian operators and Hamiltonians $\oh_{1,1}$ coincide, which is what we will prove.\\

\subsection{Double ramification hierarchy} We will not spell here the full definition of the classical double ramification hierarchy and its quantization, referring the reader to the papers \cite{Bur15, BR16a, BR16b, BDGR18}, or the review \cite{Ros17}, instead.\\

The construction has, as input, a cohomological field theory $c_{g,n}:V^{\otimes n}\to H^*(\oM_{g,n},\mbC)$ on the moduli space of stable curves $\oM_{g,n}$, with metric $\eta$ on the vector space $V$ and unit $e_1\in V$ (see \cite{KM94} for the definition) and, from its intersection theory, produces an integrable system of tau-symmetric Hamiltonian PDEs with Hamiltonian densities $g_{\alpha,d} \in \hcA_u$, $1\leq\alpha\leq N$, $d\geq -1$ where only even powers of $\eps$ appear, and Hamiltonian operator $K_u^{\alpha\beta} = \eta^{\alpha\beta} \d_x$.\\

The densities $g_{\alpha,d}$ satisfy the property $\frac{\d g_{\alpha,d+1}}{\d u^1} = g_{\alpha,d}$ with, in particular, $g_{\alpha,-1} = \eta_{\alpha\mu} u^\mu$. As remarked in \cite{BDGR16b}, an integrable system with this property and the above Poisson structure automatically possesses a tau-structure given by $h_{\alpha,d} = \frac{\delta \og_{\alpha,d+1}}{\delta u^1}$, with normal coordinates given by $\tu^\alpha = \eta^{\alpha\mu}\frac{\delta \og_{\mu,0}}{\delta u^1} $, $1\leq \alpha\leq N$.\\

In case the vector space $V$ is graded, with $\deg e_1 = 0$, the cohomological field theory is homogeneous if there exists $\delta\in \mbC$ such that the maps  $c_{g,n}:V^{\otimes n}\to H^*(\oM_{g,n},\mbC)$ and $\eta:V^{\otimes 2} \to \mbC$ have degree
$$\deg c_{g,n} = \delta (g-1), \qquad \deg \eta =-\delta,\qquad \delta \in \mbC, \quad 1\leq\alpha\leq\dim V.$$
In this case the formal variables $u^\alpha_k$ and $\eps$ acquire a grading too and the Hamiltonian densities of the resulting DR hierarchy are homogeneous:
$$|u^\alpha_k| = 1- \deg e_\alpha,\qquad | \eps| = \frac{1-\delta}{2},\qquad  |g_{\alpha,d}| = d+3-\delta-|u^\alpha|. $$
This grading $|\cdot|$ is not related to the differential grading $\deg (\cdot)$ on $\hcA$ and $\hLambda$ introduced in section \ref{section:positive flows}.\\

The DR hierarchy coincides with the Dubrovin-Zhang hierarchy \cite{DZ05,BPS12a,BPS12b} in genus $0$ (i.e. its dispersionless limit $\eps\to 0$ is the principal hierarchy of the corresponding Frobenius manifold, see \cite{DZ05}) and the two hierarchies are conjectured to be equivalent up to a Miura transformation preserving the tau-structure. This is the strong DR/DZ equivalence conjecture of \cite{BDGR18}, which has been proven for several CohFTs but is open in general. In this paper we will prove the conjecture for the $D_4$ Dubrovin-Saito CohFT.

\subsection{Integrable systems of DR type}
Let $\{\cdot,\cdot\}$ denote the Poisson structure associated to the Hamiltonian operator $\eta^{\mu \nu} \d_x$. For a local functional $\oh\in\hLambda^{[0]}_N$ consider the operator $\cD_{\oh}:\hcA[[z]]\to\hcA[[z]]$ defined by
\begin{equation*}
 \cD_{\oh}= \d_x \circ (D-1) - z\{ \cdot ,\oh\}, \qquad D:=\sum_{k\geq0}u^\alpha_k \frac{\d}{\d u^\alpha_k} + \eps \frac{\d}{\d \eps}.
\end{equation*}
Suppose there exist $N=\dim V$ solutions $g_{\alpha}(z) \in \hcA^{[0]}[[z]]$, $\alpha=1,\ldots,N$, to $\cD_{\oh} g_\alpha(z)=0$ with initial conditions $g_{\alpha}(z=0)=\eta_{\alpha\mu}u^\mu$. Then a new vector of solutions in the same class can be found by the following transformation
\begin{equation}\label{eq:solution transf classical}
g_{\alpha}(z) \mapsto a^\mu_\alpha(z) g_\mu(z) + b_\alpha(z),
\end{equation}
where $a^\mu_\alpha(z)=\delta^\mu_\alpha + \sum_{i > 0} a^\mu_{\alpha,i} z^i \in \mbC[[z]]$ and $b_\alpha(z)=\sum_{i>0} b_{\alpha,i} z^i \in \mbC[[z]]$.\\

The following result from \cite{BDGR16b} will constitute the main technical tool in our computation the DR hierarchy of the $D_4$ CohFT.\\
\begin{theorem}[\cite{BDGR16b}]\label{theorem:recursion->integrability classical}
Assume that $\oh \in \hLambda^{[0]}$ has the following properties:
\begin{itemize}
\item[(a)] there exist $N$ independent solutions $g_{\alpha}(z) = \sum_{p\geq 0} g_{\alpha,p-1} z^p\in \hcA^{[0]}[[z]]$, $\alpha=1,\ldots,N$, to the equation
\begin{equation}\label{eq:operator recursion classical}
\cD_{\oh} g_\alpha(z) = 0
\end{equation}
with the initial conditions $ g_{\alpha}(z=0)=\eta_{\alpha\mu}u^\mu$,
\item[(b)] $\displaystyle \frac{\delta \oh}{\delta u^1} = \frac{1}{2}\eta_{\mu \nu } u^\mu u^\nu + \d_x^2 r, \qquad r\in \hcA^{[-2]}$.
\end{itemize}
Then, up to a transformation of type (\ref{eq:solution transf classical}), we have
\begin{itemize}
\item[(i)] $\displaystyle g_{1,0} =  \frac{1}{2}\eta_{\mu \nu } u^\mu u^\nu + \d_x^2 (D-1)^{-1} r$,
\item[(ii)] $\og_{1,1} = \oh$,
\item[(iii)] $\displaystyle \{\og_{\alpha,p},\og_{\beta,q}\} = 0, \qquad \alpha,\beta=1,\ldots,N,\quad p,q\geq -1$,
\item[(iv)] $\displaystyle \{g_{\alpha,p},\og_{\beta,0}\} = \d_x \frac{\d g_{\alpha,p+1}}{\d u^\beta}, \qquad \beta=1,\ldots,N, \quad p\geq -1,$ 
\item[(v)] $\displaystyle \frac{\d g_{\alpha,p}}{\d u^1} = g_{\alpha,p-1}, \qquad \alpha=1,\ldots,N, \quad p\geq -1$,
\end{itemize}
hence in particular $\oh$ is part of an integrable tau-symmetric hierarchy.
\end{theorem}

We call a system of densities originating from an Hamiltonian $\oh = \og_{1,1}$ as in the theorem above an \emph{integrable system of double ramification type}. As proved in \cite{BDGR16b}, the DR hierarchy of any given cohomological field theory, is always and integrable system of DR type. This fact implies in particular that the entire hierarchy of DR Hamiltonian densities can be reconstructed from $\og_{1,1}\in \hLambda^{[0]}$ alone by means of the DR recursion equation \cite{BR16a}
\begin{equation}\label{eq:DRrecursion}
\d_x (D-1) g_{\alpha,d+1} = \{g_{\alpha,d},\og_{1,1}\}, \qquad g_{\alpha,-1}=\eta_{\alpha\mu}u^\mu, \qquad 1\leq\alpha\leq N, d\geq -1.
\end{equation}
The above result can sometimes be used to effectively compute the Hamiltonian $\og_{1,1}$ itself, starting from a limited amount of information on the CohFT, as we will show in the next section.

\subsection{$D_4$ Dubrovin-Saito CohFT}\label{section:DSCohFT}
In \cite{Dub99}, based on the work of K. Saito \cite{Sai81,Sai83a,Sai83b}, Dubrovin constructs a structure of Frobenius manifold \cite{Dub96} on the space of orbits of any finite irreducible Coxeter group. This space is biholomorphic to the space of miniversal unfoldings of the corresponding simple hypersurface singularity and the Frobenius structure is given by the Milnor ring multiplicative structure and the residue pairing at the corresponding deformation. The resulting Frobenius manifold is generically semisimple and conformal, so Givental-Teleman \cite{Giv01,Tel12} (see also \cite{PPZ15}) theory can be applied to produce a uniquely defined homogeneous cohomological field theory.\\

We are interested in this construction for the case of the $D_4$ simple singularity, $W=x^3+xy^2$. The resulting CohFT $c^{D_4}_{g,n}$ has phase space $V=\langle e_1,\ldots,e_4\rangle$ with $\deg e_1=0$, $\deg e_2=\deg e_4=\frac{1}{3}$, $\deg e_3=\frac{2}{3}$,  $\deg c^{D_4}_{g,n} = \delta = \frac{2}{3}$. In genus $0$ the corresponding Frobenius structure can described by the Frobenius potential (see for instance \cite{LRZ15})
\small
\begin{equation}\label{eq:potential}
\begin{split}
F \ =\ & \frac{ t^1 \left(t^2\right)^2}{12}+\frac{\left(t^1\right)^2 t^3}{12} -\frac{ \left(t^2\right)^3 t^3}{216}+\frac{\left(t^2\right)^2 \left(t^3\right)^3}{1296}+\frac{\left(t^3\right)^7}{1632960}+\frac{t^1 \left(t^4\right)^2}{4}
   +\frac{t^2 t^3 \left(t^4\right)^2}{24} 
   +\frac{\left(t^3\right)^3 \left(t^4\right)^2}{432},
\end{split}
\end{equation}
\normalsize
which, in particular, gives the metric
\begin{equation}
\eta = \ts \begin{pmatrix}
0 & 0 & \frac{1}{6} & 0\\
0 & \frac{1}{6} & 0 & 0\\
\frac{1}{6} & 0 & 0 & 0\\
0 & 0 & 0 & \frac{1}{2}
\end{pmatrix}
\end{equation}

Notice that this CohFT was proved to be isomorphic to the quantum singularity theory of Fan-Jarvis-Ruan-Witten  \cite{FJR07,FJR13} for the simple singularity  $W = x^3+xy^2$, with respect to the non-maximal diagonal symmetry group $\<J\> = \mbZ/3\mbZ$ \cite{FFJMR10}.

\subsection{$D_4$ double ramification hierarchy}

In \cite{GM05,FGM10} it was proved that the Dubrovin-Zhang hierarchy for the Dubrovin-Saito CohFT associated to a Coxeter group coincides with the Drinfeld-Sokolov hierarchy of the corresponding semisimple Lie algebra. As explained above, the $\eps\to 0$ limit of both the DZ and DR hierarchies of any (semisimple) CohFT coincides with the principal hierarchy of the Frobenius manifold which is, consequently, completely determined, thanks to the recursion equation (\ref{eq:DRrecursion}), by the Hamiltonian $\og^{[0]}_{1,1} = \og_{1,1}|_{\eps=0}$.\\

To compute the latter in the $D_4$ case we can use (see \cite{DZ05}) the dilaton equation $\og^{[0]}_{1,1}=\int (D-2)( F(t^*)|_{t^* = u^*}) dx$, to obtain

\begin{equation}\label{eq:dless}
\begin{split}
\og^{[0]}_{1,1}=  \int  &\left[\frac{u^1 \left(u^2\right)^2}{12} +\frac{\left(u^1\right)^2 u^3}{12} -\frac{\left(u^2\right)^3 u^3}{108} +\frac{ \left(u^2\right)^2
   \left(u^3\right)^3}{432}+\frac{\left(u^3\right)^7}{326592}+\frac{u^1 \left(u^4\right)^2 }{4} \right. \\
   & \left.+\frac{ u^2 u^3 \left(u^4\right)^2}{12}+\frac{ \left(u^3\right)^3 \left(u^4\right)^2}{144}\right]dx.
\end{split}
\end{equation}
Moreover, thanks to homogeneity of the CohFT, we know that, for any $k\geq 0$,
\begin{equation}\label{eq:degrees}
|u^1_k|=1, \qquad |u^2_k|=|u^4_k|=\frac{2}{3},\qquad |u^3_k| = \frac{1}{3}, \qquad |\eps|= \frac{1}{6},\qquad |\og_{1,1}| = \frac{7}{3}.
\end{equation}
Notice here how all the variables have positive degree, guaranteeing the polynomiality of $\og_{1,1}$ and, in fact, of all the Hamiltonian densities $g_{\alpha,d}$, $1\leq\alpha\leq N$, $d\geq -1$.

\begin{theorem}\label{thm:classicalD4}
The double ramification hierarchy for the $D_4$ Dubrovin-Saito cohomological field theory is the integrable hierarchy with Hamiltonian operator $$K_u^\DR = \begin{pmatrix}
0 & 0 & 6\d_x & 0\\
0 & 6 \d_x & 0 & 0\\
6 \d_x & 0 & 0 & 0\\
0 & 0 & 0 & 2 \d_x
\end{pmatrix}$$ and Hamiltonian densities defined by equation (\ref{eq:DRrecursion}), where
\begin{equation}\label{eq:DRHam}
\begin{split}
\og_{1,1} = \smash{\int} & \textstyle \left[ \left(\frac{1}{12} u^1 \left(u^2\right)^2+\frac{1}{12} \left(u^1\right)^2 u^3-\frac{1}{108} \left(u^2\right)^3 u^3+\frac{1}{432} \left(u^2\right)^2 \left(u^3\right)^3+\frac{1}{326592}\left(u^3\right)^7+\frac{1}{4} u^1 \left(u^4\right)^2 \right.\right. \\
   & \textstyle \left.\left.+\frac{1}{12} u^2 u^3 \left(u^4\right)^2+\frac{1}{144} \left(u^3\right)^3 \left(u^4\right)^2\right)+\left(-\frac{1}{6}
   \left(u_1^1\right)^2+\frac{1}{24} u^2 \left(u_1^2\right)^2-\frac{1}{72} \left(u^3\right)^2 \left(u_1^2\right)^2 \right.\right. \\
   & \textstyle \left.\left.-\frac{1}{108} \left(u^3\right)^2 u_1^1 u_1^3-\frac{1}{27} u^2 u^3 u_1^2
   u_1^3-\frac{1}{108} \left(u^2\right)^2 \left(u_1^3\right)^2-\frac{1}{2592}\left(u^3\right)^4 \left(u_1^3\right)^2-\frac{1}{36} \left(u^4\right)^2 \left(u_1^3\right)^2 \right.\right. \\
   & \textstyle \left.\left.-\frac{1}{4} u^4 u_1^2
   u_1^4-\frac{1}{9} u^3 u^4 u_1^3 u_1^4-\frac{1}{8} u^2 \left(u_1^4\right)^2-\frac{1}{24} \left(u^3\right)^2 \left(u_1^4\right)^2\right) \eps ^2+\left(-\frac{35}{46656} u^3 \left(u_1^3\right)^4 \right.\right. \\
   & \textstyle \left.\left.+\frac{1}{144} \left(u_1^3\right)^2 u_2^1+\frac{5}{216} u_1^2 u_1^3 u_2^2+\frac{1}{48} u^3 \left(u_2^2\right)^2+\frac{1}{54} u^3 u_2^1 u_2^3+\frac{5}{216} u^2 u_2^2
   u_2^3+\frac{7}{7776} \left(u^3\right)^3 \left(u_2^3\right)^2 \right.\right. \\
   & \textstyle \left.\left.+\frac{5}{72} u_1^3 u_1^4 u_2^4+\frac{5}{72} u^4 u_2^3 u_2^4+\frac{1}{16} u^3 \left(u_2^4\right)^2\right) \eps ^4+\left(\frac{1}{486}
   \left(u_1^3\right)^2 \left(u_2^3\right)^2+\frac{13}{5832} u^3 \left(u_2^3\right)^3-\frac{1}{112} \left(u_3^2\right)^2 \right.\right. \\
   & \textstyle \left.\left.-\frac{13}{1512} u_3^1 u_3^3-\frac{1}{972} \left(u^3\right)^2
   \left(u_3^3\right)^2-\frac{3}{112} \left(u_3^4\right)^2\right) \eps ^6+\left(-\frac{5}{1728} u_2^3 \left(u_3^3\right)^2+\frac{1}{1728}u^3 \left(u_4^3\right)^2\right) \eps
   ^8 \right. \\
   & \textstyle \left.-\frac{1}{7776}\left(u_5^3\right)^2 \eps ^{10}\right] dx
\end{split}
\end{equation}
\normalsize
\end{theorem}
\begin{proof}
As explained above we know that $g_{1,1}|{\eps=0} $ is given by equation (\ref{eq:dless}) and that, by definition, only even powers of $\eps$ appear in the Hamiltonian densities of any DR hierarchy.\\

Now remark from (\ref{eq:degrees}) that all the variables have positive degree, guaranteeing the polynomiality of $\og_{1,1}$ and, in fact, of all the Hamiltonian densities $g_{\alpha,d}$, $1\leq\alpha\leq N$, $d\geq -1$. In particular, since $|\eps|=\frac{1}{6}$ and in each monomials there are as many $x$-derivatives as powers of $\eps$, we see that any term where the power of $\eps$ is bigger than $10$ is either trivial or a total $x$ derivative.\\

One can then verify, by direct computation, that up to rescaling of the variable $\eps$, there exist a unique local functional $\og_{1,1}\in \hLambda^{[0]}$ such that: $g_{1,1}|_{\eps=0}$ is given by (\ref{eq:dless}), contains only even powers of $\eps$, $|\og_{1,1}|=\frac{7}{3}$ and $\og_{1,1}$ is of double ramification type. In practice one writes down the most general polynomial deformation of (\ref{eq:dless}) with the given degree and notices that imposing conditions (a) and (b) of theorem \ref{theorem:recursion->integrability classical} determines, up to rescalings of $\eps$, all the coefficients.\\

Finally, to determine the correct normalization of $\eps$, it is sufficient to compute the coefficient of a single monomial of $\og_{1,1}$ containing $\eps$. We can do this by recalling (see for instance \cite{BDGR18}) that, for any CohFT, $\Coef_{(u^1_1)^2 \eps^2} \og_{1,1}=-\frac{\dim V}{24}$.
\end{proof}
\vspace{0.4cm}
Recall that the DR hierarchy possesses a natural tau-structure given by the Hamiltonian densities $h_{\alpha,d} = \frac{\delta \og_{\alpha,d+1}}{\delta u^1}$. The normal coordinates for this tau structures are $\tu^\alpha =\eta^{\alpha\mu} \frac{\delta \og_{\mu,0}}{\delta u^1}$. Explicitly one obtains
\begin{equation*}\left\{
\begin{array}{l}
\tu^1 = u^1+\left(\frac{1}{36} \left(u_1^3\right)^2+\frac{1}{36} u^3 u_2^3\right) \eps ^2+\frac{1}{45} u_4^3 \eps ^4\\
\tu^2 = u^2\\
\tu^3 = u^3\\
\tu^4 = u^4
\end{array}\right.
\end{equation*}
Applying this change of coordinates one obtains the Hamiltonian operator
\small
\begin{equation*}K_\tu^\DR=
\left(
\begin{array}{cccc}
 \eps ^2\left(\frac{1}{3} \tu^3 \d_x^3+\frac{1}{2} \tu_1^3  \d_x^2+\frac{1}{6} \tu_2^3  \d_x\right) +\eps ^4 \frac{4}{15} \d_x^5 & 0 & 6 \d_x & 0 \\
 0 & 6 \d_x & 0 & 0 \\
 6 \d_x & 0 & 0 & 0 \\
 0 & 0 & 0 & 2 \d_x \\
\end{array}
\right),
\end{equation*}
\normalsize
which coincides with (\ref{eqPoissonMatrix}), and the Hamiltonian
\begin{equation*}
\begin{split}
\og_{1,1}= \smash{\int} & \textstyle \left[\left(\frac{1}{12} \tu^1 \left(\tu^2\right)^2+\frac{1}{12} \left(\tu^1\right)^2 \tu^3-\frac{1}{108} \left(\tu^2\right)^3 \tu^3+\frac{1}{432} \left(\tu^2\right)^2 \left(\tu^3\right)^3+\frac{1}{326592}\left(\tu^3\right)^7+\frac{1}{4} \tu^1
   \left(\tu^4\right)^2\right.\right.\\
   & \textstyle\left.\left.+\frac{1}{12} \tu^2 \tu^3 \left(\tu^4\right)^2+\frac{1}{144} \left(\tu^3\right)^3 \left(\tu^4\right)^2\right)+\left(-\frac{1}{6} \left(\tu_1^1\right){}^2+\frac{1}{24} \tu^2 \left(\tu_1^2\right){}^2-\frac{1}{72} \left(\tu^3\right)^2
   \left(\tu_1^2\right){}^2\right.\right.\\
   & \textstyle\left.\left.-\frac{1}{108} \left(\tu^3\right)^2 \tu_1^1 \tu_1^3-\frac{1}{27} \tu^2 \tu^3 \tu_1^2 \tu_1^3-\frac{5}{432} \left(\tu^2\right)^2 \left(\tu_1^3\right){}^2-\frac{1}{216} \tu^1 \tu^3 \left(\tu_1^3\right){}^2-\frac{1}{2592}\left(\tu^3\right)^4 \left(\tu_1^3\right){}^2\right.\right.\\
   & \textstyle\left.\left.-\frac{5}{144} \left(\tu^4\right)^2 \left(\tu_1^3\right){}^2-\frac{1}{4} \tu^4 \tu_1^2 \tu_1^4-\frac{1}{9} \tu^3 \tu^4 \tu_1^3 \tu_1^4-\frac{1}{8} \tu^2 \left(\tu_1^4\right){}^2-\frac{1}{24} \left(\tu^3\right)^2
   \left(\tu_1^4\right){}^2\right.\right.\\
   & \textstyle\left.\left.-\frac{1}{432} \left(\tu^2\right)^2 \tu^3 \tu_2^3-\frac{1}{216} \tu^1 \left(\tu^3\right)^2 \tu_2^3-\frac{1}{144} \tu^3 \left(\tu^4\right)^2 \tu_2^3\right) \epsilon ^2+\left(-\frac{1}{1458}\tu^3 \left(\tu_1^3\right){}^4+\frac{1}{144}
   \left(\tu_1^3\right){}^2 \tu_2^1\right.\right.\\
   & \textstyle\left.\left.+\frac{5}{216} \tu_1^2 \tu_1^3 \tu_2^2+\frac{1}{48} \tu^3 \left(\tu_2^2\right){}^2+\frac{1}{36} \tu_1^1 \tu_1^3 \tu_2^3+\frac{7}{7776} \left(\tu^3\right)^2 \left(\tu_1^3\right){}^2 \tu_2^3+\frac{1}{54} \tu^3 \tu_2^1
   \tu_2^3+\frac{5}{216} \tu^2 \tu_2^2 \tu_2^3+\right.\right.\\
   & \textstyle\left.\left.\frac{5}{5184} \left(\tu^3\right)^3 \left(\tu_2^3\right){}^2+\frac{5}{72} \tu_1^3 \tu_1^4 \tu_2^4+\frac{5}{72} \tu^4 \tu_2^3 \tu_2^4+\frac{1}{16} \tu^3 \left(\tu_2^4\right){}^2+\frac{1}{108} \tu^3 \tu_1^1
   \tu_3^3+\frac{1}{3888}\left(\tu^3\right)^3 \tu_1^3 \tu_3^3\right.\right.\\
   & \textstyle\left.\left.-\frac{1}{540} \left(\tu^2\right)^2 \tu_4^3-\frac{1}{270} \tu^1 \tu^3 \tu_4^3-\frac{1}{180} \left(\tu^4\right)^2 \tu_4^3\right) \epsilon ^4+\left(\frac{5}{15552} \left(\tu_1^3\right){}^2\left(\tu_2^3\right){}^2+\frac{1}{1458}\tu^3 \left(\tu_2^3\right){}^3\right.\right.\\
   & \textstyle\left.\left.-\frac{1}{112} \left(\tu_3^2\right){}^2-\frac{1}{1296}\left(\tu_1^3\right){}^3 \tu_3^3-\frac{11}{3888} \tu^3 \tu_1^3 \tu_2^3 \tu_3^3-\frac{13}{1512} \tu_3^1 \tu_3^3-\frac{1}{864} \left(\tu^3\right)^2 \left(\tu_3^3\right){}^2-\frac{3}{112} \left(\tu_3^4\right){}^2\right.\right.\\
   & \textstyle\left.\left.-\frac{7}{77760} \tu^3 \left(\tu_1^3\right){}^2 \tu_4^3-\frac{1}{2430}\left(\tu^3\right)^2 \tu_2^3 \tu_4^3+\frac{1}{135} \tu_1^1 \tu_5^3+\frac{1}{4860}\left(\tu^3\right)^2 \tu_1^3\tu_5^3\right) \epsilon ^6+\left(-\frac{55}{108864} \tu_2^3 \left(\tu_3^3\right){}^2\right.\right.\\
   & \textstyle\left.\left.+\frac{65}{54432} \tu_1^3 \tu_3^3 \tu_4^3+\frac{241}{388800} \tu^3 \left(\tu_4^3\right){}^2-\frac{1}{1620}\tu_1^3 \tu_2^3 \tu_5^3+\frac{1}{30240}\tu^3 \tu_3^3 \tu_5^3-\frac{1}{6480}\left(\tu_1^3\right){}^2 \tu_6^3\right.\right.\\
   & \textstyle\left.\left.-\frac{1}{2430}\tu^3 \tu_2^3 \tu_6^3\right) \epsilon ^8+\left(-\frac{41}{194400} \left(\tu_5^3\right){}^2+\frac{13}{68040} \tu_3^3 \tu_7^3\right) \epsilon ^{10}\right]dx
   \end{split}
   \end{equation*}
   \normalsize
which, up to a $\d_x$-exact term, agrees with (\ref{eq:DS11}).
\begin{corollary}
For the $D_4$ Dubrovin-Saito cohomological field theory, the double ramification hierarchy written in normal coordinates coincides with the Dubrovin-Zhang hierarchy which is the $D_4$ Drinfeld-Sokolov hierarchy.
\end{corollary}
\begin{proof}
We have established by direct computation that $K_\tu^\DZ = K_\tu^\DR$ and $\oh_{1,1} = \og_{1,1}$. It was proved in \cite{BG16} that this is sufficient for DZ or DR hierarchies of a CohFT with the same genus $0$ part to coincide at all genera. The identification of the DZ hierarchy of the $D_4$ Dubrovin-Saito CohFT with the $D_4$ Drinfeld-Sokolov hierarchy was proved in \cite{GM05,FGM10}.
\end{proof}

\subsection{$B_3$ and $G_2$ double ramification hierarchies}

In \cite{FFJMR10} it is proved that the $D_4$ Dubrovin-Saito CohFT and the FJRW theory of the singularity $D_4:W=x^3+xy^2$ with symmetry group $\langle J_1 \rangle$ where $J_1(x,y) = (e^{2 \pi i \frac{1}{3}}x,e^{2 \pi i \frac{1}{3}}y)$ and of the singularity $D_4^T:W=x^3y+y^2$ with symmetry group  $G_{\mathrm{max}}=\langle J_2 \rangle$ where $J_2(x,y) = (e^{2\pi i \frac{1}{6}}x,e^{2 \pi i \frac{1}{2}})$ are all isomorphic.\\

In \cite{LRZ15} it was shown that the FJRW theory for $(D_4,\langle J_1 \rangle)$ and $(D_4^T,G_{\mathrm{max}})$ carry are invariant with respect to the action of further symmetry groups $\mbZ_3$ and $\mbZ_2$ respectively.\\

Via the isomorphism and using our presentation from section \ref{section:DSCohFT} for the Dubrovin-Saito CohFT, we can express the action of the generators of these symmetries on the phase space $V$ in the following way:
\begin{equation}\label{eq:symm}
\mbZ_2: \left\{
\begin{array}{l}
e_1\mapsto e_1\\
e_2 \mapsto e_2\\
e_3 \mapsto e_3\\
e_4\mapsto -e_4
\end{array}
\right., \qquad
\mbZ_3: \left\{
\begin{array}{l}
e_1\mapsto e_1\\
e_2 \mapsto -\frac{1}{2}e_2-\frac{3}{2}e_4\\
e_3 \mapsto e_3\\
e_4\mapsto \frac{1}{2}e_2-\frac{1}{2}e_4
\end{array}\right..\end{equation}
In fact one can directly verify that the corresponding actions on the coordinates $(t^1,\ldots,t^4)$ leave the Frobenius potential (\ref{eq:potential}) unchanged and, by unicity of the Givental-Teleman reconstruction, this symmetry is inherited by the Dubrovin-Saito CohFT at all genera.\\

It was proved in \cite{LRZ15} that the restriction of a CohFT with a finite symmetry to the invariant subspace of $V$ is a partial CohFT (i.e. a system of linear maps satisfying all the axioms of a CohFT with the exception of the gluing axiom at non-separating nodes). The authors then remark that, although the Dubrovin-Zhang hierarchy cannot be directly constructed for partial CohFTs, the two restrictions of the $D_4$ Dubrovin-Zhang hierarchy to the invariant sectors with respect to the two above symmetries give well defined integrable systems, isomorphic respectively to the $B_3$ and $G_2$ Drinfeld-Sokolov hierarchies.\\

In contrast, in \cite{BDGR18} it was shown how that the double ramification hierarchy construction works for partial CohFTs too. All this then proves the following theorem.\\

\begin{theorem}
The double ramification hierarchy associated to the restriction of the $D_4$ Dubrovin-Saito CohFT to the invariant subspace with respect to the $\mbZ_2$ action described in (\ref{eq:symm}) is the integrable hierarchy with Hamiltonian operator $$K_u^\DR = \begin{pmatrix}
0 & 0 & 6\d_x \\
0 & 6 \d_x & 0 \\
6 \d_x & 0 & 0 
\end{pmatrix}$$ and Hamiltonian densities defined by equation (\ref{eq:DRrecursion}), where $\og_{1,1}$ is obtained from the $D_4$ Hamiltonian (\ref{eq:DRHam}) by imposing $u^4_k=0$. It is equivalent to the $B_3$ Drinfeld-Sokolov hierarchy.\\

The double ramification hierarchy associated to the restriction of the $D_4$ Dubrovin-Saito \hbox{CohFT} to the invariant subspace with respect to the $\mbZ_3$ action described in (\ref{eq:symm}) is the integrable hierarchy with Hamiltonian operator $$K_u^\DR = \begin{pmatrix}
0 & 6\d_x \\
6 \d_x & 0 
\end{pmatrix}$$ and Hamiltonian densities defined by equation (\ref{eq:DRrecursion}), where $\og_{1,1}$ is obtained from the $D_4$ Hamiltonian (\ref{eq:DRHam}) by imposing $u^2_*=u^4_*=0$. It is equivalent to the $G_2$ Drinfeld-Sokolov hierarchy.
\end{theorem}

\section{Quantum double ramification hierarchy for the $D_4$ Dubrovin-Saito CohFT}

\subsection{Quantum double ramification hierarchy}
In \cite{BR16b} Buryak's original definition of the integrable hierarchy associated to a CohFT was upgraded to include a quantization of the classical double ramification hierarchy. It consists of a system of quantum differential polynomials $G_{\alpha,d} \in (\hcA^\hbar)^{[\leq 0]}$, $1\leq \alpha\leq N$, $d\geq -1$ , where $\hcA^\hbar = \hcA[[\hbar]]$ and $\deg \hbar = -2$, such that $G_{\alpha,d}|_{\hbar = 0} = g_{\alpha,d}$. Similarly we define $\hLambda^\hbar = \hcA^\hbar/(\mathrm{Im} \d_x \oplus \mbC[[\eps,\hbar]])$ and the images $\oG_{\alpha,d} = \int G_{\alpha,d} dx$ of $G_{\alpha,d}$ in this quotient commute with respect to a star product that canonically quantizes the Hamiltonian operator $K^\DR = \eta \d_x$, i.e. 
$$[\oG_{\alpha,p},\oG_{\beta,q}] = 0, \qquad 1\leq \alpha,\beta\leq N, \quad p,q\geq -1$$
where, for any $f\in \hcA^\hbar$ and $\og\in \hLambda^\hbar$, $[\overline{f},\og]=\int[f,\og]dx$ and
\begin{equation}\label{eq:commutator}
\begin{split}
[f,\og]=\sum_{\substack{n\geq 1\\ r_1,\ldots ,r_n\geq 0\\ s_1,\ldots,s_n\geq 0}} \frac{(-i)^{n-1} \hbar^{n}}{n!}  &\frac{\partial^n f}{\partial u^{\alpha_1}_{s_1}\ldots \partial u^{\alpha_n}_{s_n}}  (-1)^{\sum_{k=1}^n r_k}  \left( \prod_{k=1}^n \eta^{\alpha_k \beta_k}\right) \times\\
& \times \sum_{j=1}^{2n-1+\sum_{k=1}^n (s_k+r_k)} C_j^{s_1+r_1+1,\ldots,s_n+r_n+1}  \partial_x^j  \frac{\partial^n g}{\partial u^{\beta_1}_{r_1}\ldots \partial u^{\beta_n}_{r_n}}.
\end{split}
\end{equation}
with
\begin{gather}\label{eq:relation of coefficients}
C_j^{a_1,\ldots,a_n}=
\begin{cases}
(-1)^{\frac{1}{2}(n-1+\sum a_i-j)}\tC_j^{a_1,\ldots,a_n},&\text{if $j=n-1+\sum_{i=1}^n a_i\ (\mathrm{mod}\ 2)$},\\
0,&\text{otherwise}.
\end{cases}
\end{gather}
and
\begin{gather}\label{eq:decomposition}
\prod_{i=1}^k\Li_{-d_i}(z)=\sum_{j=1}^{k-1+\sum d_i}\tC^{d_1,\ldots,d_k}_j\Li_{-j}(z), \qquad \Li_{-d}(z):=\sum_{k\ge 0}k^d z^k.
\end{gather}

In case the CohFT is homogeneous of degree $\deg c_{g,n} = \delta(g-1)$, then we have the following grading for the quantum DR hierarchy:
$$|u^\alpha_k| = 1-\deg e_\alpha, \qquad |\eps| = \frac{1-\delta}{2}, \qquad |\hbar| = 2-\delta, \qquad |G_{\alpha,d}|=d+3-\delta-|u^\alpha|$$

In light of the results of section \ref{section:DRclassical}, the DR quantization procedure can be applied to obtain a quantization of the $D_4$ Drinfeld-Sokolov hierarchy, similarly to what was achieved in \cite{BG16} for the $A_n$ Gelfand-Dickey hierarchies, for $n\le 4$.\\

\subsection{Quantum systems of DR type and the quantum $D_4$ hierarchy}

In \cite{BDGR16b} the quantum version of Theorem \ref{theorem:recursion->integrability classical} was also proved. We recall it here as, again, it will be applied to compute the $D_4$ quantum DR hierarchy. Let us consider the quantum Hamiltonian system defined by a Hamiltonian $\oH\in (\hLambda^\hbar)^{[\leq 0]}$ with respect to the standard quantum commutator introduced above. Consider the operator $\cD^\hbar_\oH:\hcA^\hbar[[z]]\to\hcA^\hbar[[z]]$ defined by
\begin{equation*}
\begin{split}
& \cD^\hbar_\oH f(z)= \d_x \circ (D-1) - \frac{z}{\hbar}[\cdot,\oH], \qquad D:=\sum_{k\geq 0} u^\alpha_k + \eps \frac{\d}{\d \eps} +2 \hbar \frac{\d}{\d \hbar}.
\end{split}
\end{equation*}
We warn the reader that we have upgraded the definition of the operator $D$. Suppose there exist $N$ solutions $G_{\alpha}(z) \in (\hcA^\hbar)^{[\leq 0]}[[z]]$, $\alpha=1,\ldots,N$, to $\cD^\hbar_\oH G_\alpha(z)=0$ with the initial conditions $G_{\alpha}(z=0)=\eta_{\alpha\mu}u^\mu$. Then a new vector of solutions can be found by the following transformation
\begin{equation}\label{eq:solution transf}
G_{\alpha}(z) \mapsto A^\mu_\alpha(z) G_\mu(z) + B_\alpha(z),
\end{equation}
where $A^\mu_\alpha(z)=\delta^\mu_\alpha + \sum_{i > 0} A^\mu_{\alpha,i} z^i \in \mbC[[z]]$ and $B_\alpha(z)=\sum_{i>0} B_{\alpha,i}(\eps,\hbar) z^i \in \mbC[[\eps,\hbar,z]]$.

\begin{theorem}[\cite{BDGR16b}]\label{theorem:recursion->integrability}
Assume that $\oH \in (\hLambda^\hbar)^{[\leq 0]}$ has the following properties:
\begin{itemize}
\item[(a)] there exist $N$ independent solutions $G_{\alpha}(z) = \sum_{p\geq 0} G_{\alpha,p-1} z^p\in (\hcA^\hbar)^{[\leq 0]}[[z]]$, $\alpha=1,\ldots,N$, to the equation
\begin{equation}\label{eq:operator recursion}
\cD^\hbar_\oH G_\alpha(z) = 0
\end{equation}
with the initial conditions $ G_{\alpha}(z=0)=\eta_{\alpha\mu}u^\mu$,
\item[(b)] $\displaystyle \frac{\delta \oH}{\delta u^1} = \frac{1}{2}\eta_{\mu \nu } u^\mu u^\nu + \d_x R + c(\eps,\hbar), \qquad R\in (\hcA^\hbar)^{[\le -1]}, \quad c(\eps,\hbar) \in \mbC[[\eps,\hbar]]$,
\item[(c)] $\oG_{1,1} = \oH$.
\end{itemize}
Then, up to a transformation of type (\ref{eq:solution transf}), we have
\begin{itemize}
\item[(i)] $\displaystyle \oG_{1,0} =  \int \left( \frac{1}{2}\eta_{\mu \nu } u^\mu u^\nu \right)dx$,
\item[(ii)] $\displaystyle [\oG_{\alpha,p},\oG_{\beta,q}] = 0, \qquad \alpha,\beta=1,\ldots,N,\quad p,q\geq -1$,
\item[(iii)] $\displaystyle \frac{1}{\hbar}[G_{\alpha,p},\oG_{\beta,0}] = \d_x \frac{\d G_{\alpha,p+1}}{\d u^\beta}, \qquad \beta=1,\ldots,N, \quad p\geq -1,$ 
\item[(iv)] $\displaystyle \frac{\d G_{\alpha,p}}{\d u^1} = G_{\alpha,p-1}, \qquad \alpha=1,\ldots,N, \quad p\geq -1$,
\end{itemize}
\end{theorem}

Again, we call a system of densities like the one described in the theorem above a \emph{quantum integrable system of double ramification type} and we recall from \cite{BDGR16b} that the quantum DR hierarchy of any CohFT is of this type. In particular one has the quantum version of equation \ref{eq:DRrecursion},
\begin{equation}\label{eq:qDRrecursion}
\d_x (D-1) G_{\alpha,d+1} = \frac{1}{\hbar}[G_{\alpha,d},\oG_{1,1}], \qquad G_{\alpha,-1}=\eta_{\alpha\mu}u^\mu, \qquad 1\leq\alpha\leq N, d\geq -1.
\end{equation}
The following theorem is readily proved by direct computation along the line of Theorem \ref{thm:classicalD4}.
\begin{theorem}\label{thm:classicalD4}
The quantum double ramification hierarchy for the $D_4$ Dubrovin-Saito cohomological field theory is the quantum integrable hierarchy with quantum commutator given by (\ref{eq:commutator}) and quantum Hamiltonian densities defined by equation (\ref{eq:qDRrecursion}), where
\begin{equation*}
\begin{split}
\oG_{1,1} = \og_{1,1} + i \hbar  \int & \left[-\frac{13  u^3 \eps ^4}{9072}+\left(-\frac{1}{540}  \left(u^3\right)^2+\frac{1}{108}  \left(u_1^3\right)^2-\frac{1}{135}  \left(u_2^3\right)^2\right) \eps ^2-\frac{1}{432}\left(u^3\right)^3 \right. \\
   & \left.+\frac{1}{144}  u^3 \left(u_1^3\right)^2-\frac{ u^1}{6}\right]\ dx
\end{split}
\end{equation*}
\end{theorem}
\begin{proof}
Since the $D_4$ Dubrovin-Saito CohFT is homogeneous with $\delta=\frac{2}{3}$, all variables have positive degrees given by (\ref{eq:degrees}) and $|\hbar|=\frac{4}{3}$. Since $|\oG_{1,1}| = \frac{7}{3}$ the quantum correction to the classical DR Hamiltonian $\og_{1,1}$ is a polynomial and imposing the conditions of theorem \ref{theorem:recursion->integrability} determine it uniquely up to a normalization constant for $\hbar$. The latter is fixed by recalling that, by definition (see \cite{BR16b}), for any CohFT we have $\Coef_{u^1 i \hbar}\, \oG_{1,1} = -\frac{1}{24}\dim V$.
\end{proof}

\end{document}